\documentclass[preprint,12pt,authoryear]{elsarticle}
 
\usepackage{geometry}
\usepackage{amsmath, amssymb, bm}

\usepackage{algorithm}
\usepackage{algpseudocode}

\usepackage{enumitem}
\usepackage[english]{babel}
\usepackage{amsthm}
\usepackage[title]{appendix}
\usepackage{caption}
\makeatletter
\newcommand{\customlabel}[2]{%
   \protected@write \@auxout {}{\string \newlabel {#1}{{#2}{\thepage}{#2}{#1}{}} }%
   \hypertarget{#1}{#2}
}
\makeatother
\usepackage{bbm}
\usepackage{subcaption}
\usepackage{hyperref}
\usepackage{float}
\usepackage{comment}
\theoremstyle{definition}
\newtheorem{proposition}{Proposition}[section]
\newtheorem{remark}{Remark}[section]

\newcommand{\K}{\mathbf{K}}
\newcommand{\C}{\bm{C}}

\newcommand{\Y}{\bm{y}}
\newcommand{\ppsi}{\bm{\psi}}

\newcommand{\CC}{\mathrm{C}}
\newcommand{\CCC}{\mathbf{C}}
\newcommand{\YY}{\mathbf{Y}}
\newcommand{\ZZ}{\mathbf{Z}}

\newcommand{\Z}{\bm{z}}
\newcommand{\A}{\mathcal{A}}
\newcommand{\B}{\mathcal{B}}
\newcommand{\SSS}{\mathcal{S}}

\newcommand{\s}{\bm{s}}
\newcommand{\mbQ}{\mathbf{Q}}

\newcommand{\kl}{Karhunen-Lo\`{e}ve }
\newcommand{\kle}{Karhunen-Lo\`{e}ve expansion }

\newcommand{\supvec}{^{\operatorname{vec}}}
\usepackage{natbib}
\usepackage{amsfonts, amsmath}
\usepackage{bm, bbm}
\usepackage{graphicx}
\usepackage{bbm}
\usepackage{color}
\newcommand{\x}{\bm{x}}

\newcommand{\cov}{\operatorname{cov}}
\newcommand{\trace}{\operatorname{trace}}
\newcommand{\var}{\operatorname{var}}
\newcommand{\diag}{\operatorname{diag}}

\newcommand{\ar}{^A}
\begin{document}

\begin{frontmatter}

\title{{A Criterion for Aggregation Error for Multivariate Spatial Data}}
\author[1]{Ranadeep Daw \footnote{Corresponding author. Email: rdaw@uwf.edu}}
\author[2]{Jonathan R. Bradley}
\author[3]{Christopher K. Wikle}
\author[3,4]{Scott H. Holan}
\affiliation[1]{
Department of Mathematics and Statistics, University of West Florida, Peensacola, FL, USA}
\affiliation[3]{Department of Statistics, University of Missouri, Columbia, MO, USA }
\affiliation[2]{Department of Statistics, Florida State University, Tallahassee, FL, USA }
\affiliation[4]{Research and Methodology Directorate, U.S. Census Bureau, {Washington, D.C.,} USA}

\begin{abstract}
{The criterion for aggregation error (CAGE) is an important metric that aims to measure errors that arise in multiscale (or multi-resolution) spatial data, referred to as the modifiable areal unit problem and the ecological fallacy. Specifically, CAGE is a measure of between scale variance of eigenvectors in a Karhunen-Lo\'{e}ve expansion (KLE), motivated by a theoretical result, referred to as the ``null-MAUP-theorem,'' that states that the MAUP/ecological fallacy are not present when this variance is zero. CAGE was originally developed for univariate spatial data, but its use has been applied to multivariate spatial data without the development of a null-MAUP-theorem in the multivariate spatial setting. To fill this gap, we provide theoretical justification for a multivariate CAGE (MVCAGE), which includes multiscale multivariate extensions of the KLE, Mercer's theorem, and the-null-MAUP theorem. Additionally, we provide technical results that demonstrate that the MVCAGE is preferable to spatial-only CAGE, and extend commonly used basis functions used to compute CAGE to the multivariate spatial setting. Empirical results are provided to demonstrate the use of MVCAGE for uncertainty quantification and regionalization}
\end{abstract}


\end{frontmatter}

\section{Introduction}
%

{The Modifiable Areal Unit Problem (MAUP) and the ecological fallacy can be interpreted as a manifestation of Simpson's Paradox in the spatial statistics setting. In particular, the MAUP occurs when inferential conclusions based on data aggregated to one level, say counties, differs from the conclusions using data aggregated to a different support, say census tracts. The ecological fallacy is similar, which occurs when point-referenced data produces different conclusions than aggregated data. The literature is well-aware of these issues with work dating back to Gehike and Biehl (1934) and 	\citep{openshaw1979million}, and has since become a common consideration \citep[e.g., see][]{gotway2002combining, banerjee2004hierarchical, wikle2005combining,robinson2009ecological, bradley2017regionalization,zhou2023bayesian} and is covered in standard textbooks \citep{cressie1993statistics,cressie2015statistics,waller2004applied,banerjee2004hierarchical}. The reason for this focus in the literature, is that assessing these types of spatial aggregation error is critically important for producing meaningful subject matter conclusions. 

	
Accounting for the MAUP in a likelihood framework is particularly difficult. \citet{bradley2017regionalization} introduced a metric that measures the degree at which the MAUP is present referred to the Criterion for AGgregation Error (CAGE). The CAGE is motivated by a theoretical result that states that no MAUP (or ecological fallacy) will be present provided that the eigenvectors are piecewise constant over the areal units within the spatial support \citep[e.g, see Theorem ][Proposition~2]{bradley2017regionalization}, and we refer to result as the ``null-MAUP-theorem.'' The CAGE has been used in a variety of context, such as boundary detection \citep{qu2021boundary} and multivariate spatial settings \citep{daw2022overview,zhou2023bayesian}. However, there are number of practical issues when implementing CAGE limiting its broad use. In particular, the CAGE requires one to orthogonalize the implied basis functions in the spatial model, which requires the Cholesky decomposition of a matrix that is not numerical stable and does not have a guarantee to exist. Additionally, the original theoretical motivations for the CAGE have not been extended to the multiscale multivariate spatial setting, despite its use in these contexts.

Building on the results from \citet{bradley2017regionalization}, we provide non-trivial theoretical development of the CAGE in the multivariate spatial setting. Extending CAGE to the multivariate spatial setting requires the development of new multiscale-multivarite spatial versions of the Karhunen-Lo\'{e}ve expansion (KLE), Mercer's theorem, and the null-MAUP-theorem. These results are distinctly different from \citet{bradley2017regionalization}, which only provides KLE and Mercer's theorem extensions for the multiscale univariate spatial setting. Similarly, \citet{daw2022overview} discuss extensions for a uniscale multivariate spatial KLE and Mercer's theorem, but does not provide a null-MAUP-theorem. Providing this theory is quite important as it gives justification to the use of a MultiVariate CAGE (MVCAGE), which is a key contribution this article. MVCAGE utilizes the multivariate equivalent of the KLE \citep[see][for a review of multivariate KLE and the initial introduction to the MVCAGE]{daw2022overview}, where, motivated by our novel null-MAUP-theorem, the between scale variability of multivariate eigenfunctions are used to quantify the MAUP/ecological fallacy. Thus, given a set of areal units, MVCAGE acts as a measure of uncertainty of multiscale errors that arise in multiscale multivariate spatial data.

In addition to our new multiscale multivariate KLE, Mercer's theorem, and null-MAUP theorem, we explore the role of incorporating multivariate dependence theoretically. In particular, MVCAGE is an estimate of the between-scale-variance, and hence, a natural question to ask is whether or not incorporating multivariate dependence aids with estimating the ``true between-scale-variance.'' This is an important fourth contribution to the CAGE literature, as we provide a technical result that shows that one more precisely estimates the true between-scale-variance when leveraging multivariate spatial dependence as opposed to only leveraging spatial-only dependence.

 Spatial change of support (COS) has become a common inferential consideration in the multiscale spatial setting , which refers to the problem of performing inference (i.e., commonly spatial prediction) on a set of regions (called the target support) that differs from the spatial support the data is observed on (called the source support) \citep[e.g., see][for standard references]{waller2004applied,cressie1993statistics,cressie2015statistics,banerjee2004hierarchical}. COS allows us to consider several different target supports for prediction and estimation, and MVCAGE provides a metric to optimize to choose competing target supports. Thus, in addition to uncertainty quantification, the MVCAGE also provides a metric to use for regionalization, which is the general problem of choosing a rarget support. Univariate regionalization has been widely studied and used in geography and other disciplines to identify regions with similar values of a single variable \citep{openshaw1984modifiable, anselin1988spatial, bailey1995interactive}. In particular, given a collection of competing spatial supports, one can choose the support that minimizes MVCAGE effectively minimizing the MAUP and ecological fallacy similar to that of \citet{bradley2017regionalization} and \citet{qu2021boundary}. In this new multivariate spatial context, we consider a single optimal regionalization across all variables, as this allows one to naturally compare predictions across variables. 

The use of MVCAGE for uncertainty quantification and regionalization requires one to specify eigenfunctions that define the latent process in their statistical model.  \citet{bradley2017regionalization} consider the Obled-Creutin (OC) basis functions \citep{obled1986some} to the multiscale spatial setting, and \citet{daw2022overview} considered its use in the univariate spatial context. One of the original motivations for the development of OC basis functions is that the OC basis function can be easily adapted to several different basis function expansion models, allowing one to make use of their favorite basis set while achieving a complete orthogonal basis set. Specifically, the OC basis function starts with any given (possibly non-orthogonal) basis set, referred to as ``generating basis functions,'' and orthogonalize them in a manner that respects the Fredholm integral equations, so that the resulting OC basis functions can be interpreted as complete via the KLE. 

 The re-normalization of the generating basis functions that define the OC basis functions, can be written as the product of two terms: (1) a Cholesky matrix formed by integration of the generating basis functions, and (2) an unknown orthogonal matrix. One difficulty with this paramaeterization is that the Cholesky matrix in Item 1 does not necessarily exist for non-orthogonal basis sets. To avoid this issue, one can restrict the generating basis functions to be orthogonal (e.g., Fourier basis functions, Hermite basis functions, wavelets, etc.), so that the matrix in Item 1 is simply the identity matrix. Additionally, Item 2 can be specified to be the eigenvectors of a pre-defined covariance matrix. This special case of orthogonal generating basis function produces an OC basis set that is equivalent to the basis functions used in \citet{cho2013karhunen} and \citet{happ2018multivariate}. 
}

The manuscript is structured as follows. Section~\ref{sec2} provides preliminary details including notation, and background on the univariate and multivariate KLEs. Section~\ref{sec3} presents our novel theoretical results including a result that shows that multivariate spatial CAGE outperforms spatial-only CAGE in terms of mean squared error, the multiscale multivariate spatial extensions to the KLE, Mercer's theorem, the null-MAUP-theorem, and the construction of the MVOC basis functions. Section~\ref{sec4} showcases improvements in uncertainty quantification and along with regionalization via simulations and an example of spatial change of support. Section~\ref{sec5} discusses the advantages and limitations of our approach, and puts forward ideas for future research.

\section{Preliminary Details} \label{sec2}
{In this section, we present the mathematical notation used in this paper (Section~\ref{sec:notation}) and the background reviews of the univariate (Section~\ref{sec:univ}) and multivariate KLEs (Section~\ref{sec:smkle}).}

\subsection{Notation}\label{sec:notation}
Let $\SSS \in \mathbb{R}^d$ denote the spatial domain under study, where $d$ is often $1$, $2$, or $3$ for geospatial problems. We consider a set of $N$ univariate spatial random variables $\{Z_1(\s), \ldots, Z_N(\s) \}: \SSS \to \mathbb{R}$. Corresponding to each $Z_j(\s)$, we define the latent, noise-free, zero-mean, continuous spatial process by $Y_j(\s): \SSS \to \mathbb{R}, j = 1, \ldots, N$. We denote the mean function of $Z_j(\s)$ as $\mu_j(\s)$, and also define a set of error processes $\epsilon_j(\s): \SSS \to \mathbb{R}$ with zero-mean and finite variances for $j = 1, \ldots, N$.

In the multivariate context, we use lower-case bold font to denote the respective multivariate versions. For example, $\Z(\s) = \big(Z_1(\s), \ldots, Z_N(\s) \big)^{\top} \in \mathbb{R}^N$ denotes the multivariate spatial observations, $\Y(\s) = \big(Y_1(\s), \ldots, Y_N(\s) \big)^{\top} \in \mathbb{R}^N$ shows the corresponding vector valued latent process, and $\bm{\mu}(\s) = \big(\mu_1(\s), \ldots, \mu_N(\s) \big)^{\top} \in \mathbb{R}^N$ is the vector-valued mean function. We assume that $\Z(\s)$ is observed at $n$ locations denoted by $\bm{S} = \{\s_1, \ldots, \s_n \}$. We allow for missing observations, i.e., we do not require each $Z_j$ to be observed at every $\s_k \in \bm{S}$. The math-bold symbols (e.g., $\ZZ$, $\YY$) denote the data matrices corresponding to the bold-faced symbols. For example, the $n \times N$ matrix $\ZZ$ is the observed data, where the $j,k$-th element of $\ZZ$ is the realization of the $k$-th univariate variable $Z_k$ at location $\s_j$ (i.e., $Z_k(\s_j)$).

Next, we introduce covariance functions and covariance matrices. Given two spatial locations $\s$ and $\bm{r}$, the positive-definite covariance function of the univariate process $Y_j(\cdot)$ is denoted as $\CC_{jj}(\s, \bm{r}) = \cov\big[Y_j(\s), Y_j(\bm{r}) \big]: \SSS \times \SSS \to \mathbb{R}$. The corresponding covariance matrix is the evaluation of $\CC_{jj}$ at $\bm{S} \times \bm{S} = \{ (\s_k, \s_{\ell}) : k, \ell =1, \ldots, n \}$;  is denoted by $\C_{jj} \in \mathbb{R}^{n \times n}$, and its $k, \ell$-th element is given by $\cov\big[Y_j(\s_k), Y_j(\s_{\ell}) \big]$. The inter-variable cross-covariance functions are of a similar form given by $\CC_{ij}(\s, \bm{r}) = \cov \big[Y_i(\s), Y_j(\bm{r}) \big]$. We similarly denote the cross-covariance matrices (i.e., evaluations of $\CC_{ij}$ over $\bm{S} \times \bm{S}$) as $\C_{ij} \in \mathbb{R}^{n \times n}$. The multivariate covariance function of the multivariate spatial process $\Y$ is denoted by $\mathbb{C}(\cdot, \cdot): \SSS \times \SSS \to \mathbb{R}^{N \times N}$. We similarly denote the multivariate joint covariance matrix by $\CCC$. It is easy to see that $\CCC$ is a block matrix of order $\mathbb{R}^{nN \times nN}$ with blocks $\C_{ij} \in \mathbb{R}^{n \times n}; \, i,j = 1, \ldots, N$. Now, for the $j$-th univariate covariance function $\CC_{jj}$, we denote its $k$-th eigenfunction as $\psi_{jk}(\s) : \SSS \to \mathbb{R}$ and the corresponding eigenvalue is denoted by $\lambda_{jk}$ for $k=1, 2, \ldots, \infty$.

To consider spatial data at different resolutions, we often use the terms ``point-level'' and ``areal'' (or area-level). Point-level or point-referenced spatial data refers to random variables whose arguments are the points from $\SSS$, i.e., subsets of $\SSS$ with Lebesgue measure zero. Alternatively, areal data denotes spatial random variables that take non-zero Lebesgue measurable subsets of $\SSS$ as arguments. We use bold lowercase vectors (e.g., $\s$, $\bm{r}$) to denote the point locations and calligraphic capital letters (e.g., $\A$, $\B$, $\SSS$) to denote areal units. Area-level random variables are denoted by a superscript ``$A$''. For example, $\Z\ar(\A) = \big(Z_1\ar(\A),  \ldots, Z_N\ar(\A) \big)^{\top}$ and $\Y\ar(\A)  = \big(Y_1\ar(\A),  \ldots, Y_N\ar(\A) \big)^{\top}$ are the multivariate areal data and process vectors. Without any superscript, the random variables (e.g., $Z_j(\s)$, $\Y_k(\s)$) denote the point-level versions. Note that the area of $\SSS$ can be expressed in terms of the point-level support as $|\SSS| = \int_{\s \in \SSS} \, d\s$ and areal support as $\SSS  = \sum_{\A_j \in \widetilde{\A}} |\A_j|$ where $ \widetilde{\A} = \{ \A_1, \A_2, \ldots, \A_M:  \cup_j \A_j = \SSS, \A_i \cap \A_j = \emptyset \}$.

It is important to note that we may need to deal with multiple area-level resolutions. For example, one may want to join neighboring census block-groups to create a census tract or county. In such cases, we use $\B$ to represent the areal units at the higher resolution and $\A$ to denote the aggregated lower resolution. In this example, $Z_j^{\B}$ corresponds to a variable at the census block-group level, and $Z_j\ar$ are the created census tracts or counties.

\subsection{Review: Univariate \kl Expansion}\label{sec:univ}
For each of the latent univariate processes $Y_j$ with covariance function $\CC_{jj}$, there exists a KLE of the following form \citep{Karhunen1947, Loeve1945, cressie2015statistics}
\begin{equation} \label{ekle}
    Y_j(\s) = \sum_{k=1}^{\infty} \alpha_{jk} \psi_{jk}(\s),
\end{equation}
{in mean squared error (or $\mathcal{L}_{2}$).} Here, the expansion functions $\psi_{jk}(\s): \SSS \to \mathbb{R}$ are deterministic, spatially-varying functions, and the expansion coefficients $\alpha_{jk}$'s are mean-zero, uncorrelated random variables with respective variances $\lambda_{j1} \geqslant \lambda_{j2} \geqslant \cdots \geqslant 0$. The KLE is the consequence of Mercer's theorem \citep{aronszajn1950theory}, which {leads to the expression in (\ref{ekle}) in $\mathcal{L}_{2}$}. Together with Mercer's theorem, the KLE yields the following properties:
\begin{align*}
    \textbf{(P1:) } \qquad &\CC_{jj}(\s, \bm{r}) = \sum_{k=1} ^{\infty} \lambda_{jk} \psi_{jk}(\s) \psi_{jk}(\bm{r})\,, \\ 
    \textbf{(P2:) } \qquad  &\int_{\mathcal{S}} \psi_{jk}(\s) \psi_{j\ell}(\bm{s}) \, d\s = \delta_{k \ell} \qquad \text{(i.e., $1$ if $k = \ell$, $0$ otherwise)}\,, \\
    \textbf{(P3:) } \qquad  &\mathbb{E}[\alpha_{jk}] = 0 \,, \\
    \textbf{(P4:) }  \qquad &\cov \big[ \alpha_{jk}, \alpha_{j\ell}\big] = \lambda_{jk} \delta_{k \ell} \,, \\
    \textbf{(P5:) }  \qquad &\alpha_{jk} = \int_{\mathcal{S}} Y_{j}(\s)\psi_{jk}(\s) \, d\s \,,
\end{align*}
where $\delta_{k\ell}$ represents the Dirac delta function. The above properties demonstrate the bi-orthogonal property of the KLE, which means that the KLE uses both the uncorrelated expansion coefficients and orthonormal eigenfunctions in the expansion in \eqref{ekle}. 

The eigenvalues $\lambda_{j1}, \lambda_{j2}, \ldots$ are non-negative for any positive-definite covariance function and monotonically decrease to zero. Since $\var(\alpha_{jk}) = \lambda_{jk} \overset{k \to \infty}{\to} 0$, the random variables ($\alpha_{jk}$) also shrink to the Dirac zero distribution. Therefore, this provides the rationale behind the truncation of the infinite sum in Equation \eqref{ekle} at some finite number $M_j$. Given such a cutoff value $M_j$, the truncated KLE satisfies the optimal $\mathcal{L}_2$ approximation criterion, i.e., $\widetilde{Y}_j(\s) = \sum_{k=1}^{M_j} \alpha_{jk} \psi_{jk}(\s)$ has minimum mean-square prediction error for $Y_j(\s)$ among all linear expansions with $M_j$ terms. This optimal approximation, in addition to the bi-orthogonality criterion, justifies using the KLE as an optimal linear expansion in analyzing dependent processes. 

\begin{remark}
Note that the choice of $M_j$ is often data dependent and usually chosen by the practitioner. There are some standard techniques to find a reasonable number of eigenpairs, such as scree plots, percentage of variance explained, and pre-fixed cutoff value, etc. In practice, the final choice of $M_j$ should include a sensitivity analysis.
\end{remark}

\subsection{Review: Multivariate KLE} \label{sec:smkle}
In this section, we discuss the extension of the KLE to multivariate domains to accommodate multivariate spatial processes. This is fundamental to our multivariate CAGE criterion. To describe the multivariate KLE (MKLE), we consider a multivariate spatial process $\Y(\s): \SSS \to \mathbb{R}_N$ with mean $\bm{0}$ and finite covariance function $\mathbb{C}(\s, \bm{r})$ and from the multivariate version of Mercer's theorem {and the KLE \citep[e.g., see][]{daw2022overview} we have the following properties}:
\begin{align} \label{eqmkle} 
   { \textbf{(P6:) }}  \qquad &\mathbb{C}(\s, \bm{r}) = \sum_{k=1} ^{\infty} \lambda_{k} \ppsi_{k}(\s) \ppsi_{k}(\bm{r})^{\top},  \\ \nonumber
    {\textbf{(P7:) } } \qquad &\Y(\cdot) = \sum_{k=1}^{\infty} \alpha_k \ppsi_{k}(\cdot){,\hspace{3pt}\mathrm{in}\hspace{3pt}\mathcal{L}_{2}}\\  \nonumber
    {\textbf{(P8:) } } \qquad &\int_{\mathcal{S}} \ppsi_{k}^{\top}(\s) \ppsi_{\ell}(\bm{s}) \, d\s = \delta_{k \ell}, \\ \nonumber
    {\textbf{(P9:) } } \qquad &\mathbb{E}[\alpha_k] = 0, \\ \nonumber
   { \textbf{(P10:) }}  \qquad &\cov \big[ \alpha_k, \alpha_{\ell}\big] = \lambda_k \delta_{k \ell}. \nonumber
\end{align}
Here, the expansion coefficients $\alpha_k$-s are the stochastic component of the MKLE and the corresponding multivariate eigenfunctions $\ppsi_k$ are the deterministic expansion functions. Each $\alpha_k$ is a univariate mean-zero random variable with variance $\lambda_k$, the $k$-th eigenvalue of the multivariate covariance function $\mathbb{C}(\cdot, \cdot) \to \mathbb{R}^{N \times N}$. Note that $\mathbb{C}(\s, \mathbf{r})$ is matrix-valued, symmetric, and positive-definite with monotonically shrinking eigenvalues $\lambda_1 \geqslant \lambda_2 \geqslant \ldots \geqslant 0$ and multivariate orthonormal eigenfunctions $\ppsi_k(\s): \SSS \to \mathbb{R}^N, k = 1,2, \ldots$ of $\mathbb{C}$ as in the above. Similar to the univariate case, $\alpha_k$-s shrink to the Dirac zero distribution as $k \to \infty$.

Although the above formulation is intuitive from the univariate KLE, the challenge lies in defining the multivariate covariance functions and the corresponding vector-valued eigenfunctions $\ppsi_k$. In practice, one should be careful in maintaining the positive-definiteness of $\mathbb{C}$. The vector-valued eigenfunctions and their orthonormality constraints can also be difficult to work with. As such, it is easier to build these multivariate eigenfunctions from their process-specific univariate KLEs. We provide a brief description of this procedure next.

Proposition $5$ from \citet{happ2018multivariate} establishes the bijective relationship between a multivariate KLE and its process-specific univariate KLEs. {That is, the elements of $\Y(\cdot)$ can be represented via univariate KLEs.} Suppose we have the univariate KLEs in the form of \eqref{ekle} for all the $N$ univariate processes. For the $j$-th process, the expansion coefficients $\alpha_{jk}$ and $\alpha_{j\ell}$ are uncorrelated for $k \neq \ell$. However, for two different processes $i$ and $j$, the expansion coefficients $\alpha_{ik}$ and $\alpha_{j\ell}$ are correlated. We decorrelate the expansion coefficients in the following manner. Define $\bm{K}_{ij}$ as the $M_i \times M_j$ matrix with $k, \ell$-th element given by $K^{ij}_{k\ell} = \cov\big[ \alpha_{ik}, \alpha_{j\ell}\big]$. From the Fredholm integral equations, we have the following  \citep{cho2013karhunen}
\begin{align*}
    K^{ij}_{k\ell} = \int_{\SSS} \int_{\SSS} \CC_{ij}(\s, \bm{r}) \psi_{ik}(\s) \psi_{j\ell}(\bm{r}) \, d\s \, d\bm{r}, \qquad \text{for $k, \ell= 1, \ldots, n$.}
\end{align*}
We denote $\K$ as the block matrix with $\bm{K}_{ij}$ as the $i,j$-th block. Each $\bm{K}_{jj}$ is a diagonal matrix with the eigenvalues of the $j$-th process as the diagonal elements. Note that $\K$ is a symmetric matrix of size $\mathbb{R}^{M \times M}$, where $M = \sum_{j=1}^N M_j$. Assuming the positive definiteness of $\K$, we must compute the eigenvalues and eigenvectors of $\K$, through which we decorrelate the univariate expansion coefficients $\alpha_{jk}$-s of the different univariate processes. Proposition $5$ from \citet{happ2018multivariate} proves that the eigenvalues of $\K$ are the eigenvalues for the MKLE of $\mathbb{C}$. 

{To implement the multivariate KLE from \citep{happ2018multivariate} in practice, one can truncate the multivatiate KLE to $M$ terms as follows.} Denote the $k$-th eigenvector of $\K$ as $\bm{e}_k$, where $\bm{e}_k \in \mathbb{R}^M$. Each $\bm{e}_k$ can be written in the blocked form with blocks $\bm{e}_k^{1}, \ldots, \bm{e}^N_k$. If the $j$-th univariate truncated KLE of $\CC_{jj}$ has $M_j$ terms, the $j$-th block $\bm{e}^{j}_k$ is of size $M_j$. Then, the $j$-th element of the $k$-th {eigenvector of the} truncated MKLE is given by 
\begin{align} \label{cons1}
    \big[ \ppsi_{k} \big]_j = \big( \psi_{j1}, \, \ldots, \,  \psi_{jM_j} \big) \bm{e}_k^j.
\end{align}
\noindent Similarly, the expansion coefficients for the MKLE are then computed as
\begin{align} \label{cons2}
    \alpha_k = \sum_{j=1}^{N} \big( \alpha_{j1}, \, \ldots,  \,\alpha_{jM_j} \big) \bm{e}_k^j,
\end{align}
If the truncated KLE of $Y_j$ has $M_j$ terms, the vector-valued eigenfunctions are of dimension $M = \sum_{j=1}^N M_j$. Similar to the univariate case, the MKLE yields the minimum $\mathcal{L}_2$ error (i.e., mean-square error) among all linear expansions of $\Y$ with $M$-terms, which makes the MKLE the $\mathcal{L}_2$-optimal expansion method along with the bi-orthogonality property of uncorrelated expansion coefficients and orthonormal expansion functions. { From Proposition 6 in \citep{happ2018multivariate} this truncated MKLE converges to $\Y(\cdot)$ in $\mathcal{L}_{2}$. We refer to this specification of eigenfunctions as the Happ and Greven (HG) basis functions.}

\section{Novel Theoretical Developments for MVCAGE}
{In Section \ref{sec:mmkle} we provide the extension to the multiscale multivariate KLE and Mercer's theorem. Then, in Section~\ref{sec:MVCAGE}, we provide the null-MAUP theorem and demonstrate that MVCAGE is preferable to spatial-only CAGE. Finally we develop the MVOC basis functions in Section~\ref{klecmp}.}

\subsection{The Multiscale Multivariate Spatial Karhunen-Lo\'{e}ve Expansion and Mercer's Theorem}\label{sec:mmkle}
Using the MKLE for dependent data, we now address the issue of the spatial COS and propose the {multiscale multivariate KLE and Mercer's theorem. Consider a} point-level (multivariate) spatial process $\Y(\s): \SSS \to \mathbb{R}^N$. We define the corresponding areal-version of this variable {via COS} as
\begin{align}  \label{eqcos}
    \noindent  \Y\ar(\A) &{\equiv} \frac{1}{|\A|}\int_{\s \in \A} \Y(\s) \, d\s\\
    &  {= \left(\frac{1}{|\A|}\int_{\s \in \A} Y_{1}(\s) \, d\s,\ldots,\frac{1}{|\A|}\int_{\s \in \A} Y_{N}(\s) \, d\s\right)^{\top}}\\
    & {\equiv (Y\ar_{1}(\A),\ldots, Y\ar_{N}(\A))^{\top}},
\end{align}
\noindent where $\Y\ar(\A)$ are the continuous average of $\Y(\s)$ over an areal unit $\A$  \citep[e.g., see][]{cressie2015statistics}. {Similarly, define 
	\begin{align}  \label{eq:bfcos}
		\noindent  \ppsi\ar_{k}(\A) &{\equiv} \frac{1}{|\A|}\int_{\s \in \A} \ppsi_{k}(\s) \, d\s,
	\end{align}
	where $\ppsi_{k}$ is the HG basis function. This is a traditional formula for spatial COS applied to a multivariate spatial process. To date there is no KLE and Mercer's theorem for multiscale multivariate areal-reference spatial data. While the expressions of the multiscale multivariate KLE and Mercer's theorem are very similar to that of the multiscale univariate versions and uniscale multivariate versions, we stress that these existing theorems can not immediately be applied to the multiscale multivarite spatial setting and the extension to this setting requires careful consideration.\\

\begin{proposition}\label{acov}
	 Let $(\Omega, \mathcal{B},\mathcal{P})$ be a probability space, where $\Omega$ is a sample space, $\mathcal{H}$ is a $\sigma$-algebra on $\Omega$ and $\mathcal{P}$ is a finite Borel measure. Let $\mathcal{P}_{U}(\A)$ be the measure associated with a uniform distribution on $\A\subset \mathbb{R}^{d}$ and assume that $\mathcal{P}\times \mathcal{P}_{U}(\A)$ defines a $\sigma$-finite product measure on $\Omega\times \A$.  Let $Y_{j}$ be a zero mean spatial process defined by the mapping $Y :\mathcal{S}\times \Omega\rightarrow \mathbb{R}$ for $j = 1,\ldots, N$, such
	that $Y$ is measurable for every $\s\in \mathcal{S}$, and $ \mathcal{S}\subset\mathbb{R}^{d}$ is a topological Hausdorff space. Let the covariance function $C_{jj}(\s,\bm{r})$ be valid, continuous, and exists for all $\s, \bm{r} \in \mathcal{S}$. Let $\mathcal{L}_{2}(\Omega)$ denote the Hilbert space of real-valued square integrable random variables. Additionally assume the native Hilbert space for the multivariate covariance function $\mathbb{C}$ is separable, and the elements of $\int \mathbb{C}(\s,\s)d\s$ are bounded.
	\begin{enumerate}[label=\alph*.]
		\item \textbf{multiscale multivariate KLE:} It follows that,
		\begin{align*}
			\Y\ar(\A) &= \sum_{k=1}^{\infty} \alpha_k \ppsi\ar_{k}(\A)
		\end{align*} 
		in $\mathcal{L}_{2}(\Omega)$, where $\alpha_{k}$ are defined in (\ref{cons2}) and $\ppsi\ar_{k}$ is defined in (\ref{eq:bfcos}).
		\item \textbf{multiscale multivariate Mercer's theorem:}  It follows that,
\begin{align*}
	\cov \big[\Y\ar(\A_i), \Y\ar(\A_j) \big] = \sum_{k=1}^{\infty} \lambda_k \,\ppsi\ar_{k}(\A_i) \, \ppsi\ar_{k}(\A_j)^{\top},
\end{align*}
for $\A_{i},\A_{j} \in \mathcal{S}$.
	\end{enumerate}
\end{proposition}
\noindent
\textbf{Proof:} See Appendix A.\\

\noindent
The conditions for Proposition 3.1 are mostly standard conditions for traditional KLE theory \citep[e.g., see][for more details]{daw2022overview}, with the added condition of a product measure to interchange expectations with respect to $\Y$ and COS via Fubini's theorem. 

Representation theorems such as the KLE, Mercer's theorem, the measurability theorem \citep{resnick2013probability}, stick-breaking theorem \citep{sethuraman1994constructive}, and the Kolmogorov-Arnold theorem \citep[e.g., see][for a more recent discussion]{schmidt2021kolmogorov} among many others,  play a crucial role in semi-parameteric Bayesian inference. In particular, the multivariate process $\Y(\cdot)$ is allowed to follow its true generating mechanism and such theorems are known to approximate the random process arbitrarily well (i.e., in $\mathcal{L}_{2}$ in our case) with fewer concerns of model misspecification as a result. This particularly pertinent to the multivariate areal-referenced setting, which is an sub-domain in statistics that tends to adopt strong parametric assumptions such as the linear model for coregionalization \citep[LMCl][]{journel1978mining, goulard1992linear} and the multivariate Mat\'{e}rn \citep{gneiting2010matern}. Consequently, this strategy to use complete (in the sense of Proposition 3.1) areal-referenced basis functions offers an important contribution to areal-referenced-only and multiscale multivariate spatial processes.

}

\subsection{The Multivariate CAGE}\label{sec:MVCAGE}
{
The null-MAUP-theorem uses the KLE to provide insights on when there are no concerns of the MAUP or the ecological fallacy. This theorem has only been explicitly developed in the multiscale univariate setting, leading to our next contribution.

\begin{proposition} \label{prop1}
Adopt the same assumptions as in Proposition (\ref{acov}). Consider any continuous functional $\bm{f}: \mathbb{R}^{n_{A}}\times \mathbb{R}^{N} \to \mathbb{R}^K$. Let $\lambda_{k}>0$ for all $k$. Consider the multivariate spatial process at three different resolutions: a point level support $\{\bm{x}_{j}: j =1,\ldots, n_{A}\}$, $\B_{1},\ldots, \B_{n_{A}}$, $\A_{1},\ldots, \A_{n_{A}}$, such that $\bm{x}_{j}\in \B_{j}\in\A_{j}$ for all $j$. Define the $n_{A}\times N$ matrix $\textbf{Y}^{\A}$ with $j$-th row $\Y^{\A}(\A_{j})$. Similarly define $\textbf{Y}^{\B}$ and $\textbf{Y}^{\bm{x}}$.
	\begin{enumerate}[label=\alph*.]
		\item \textbf{null-ecological-fallacy-theorem:} A necessary and sufficient condition for $\bm{f}(\textbf{Y}^{\bm{x}}) \overset{a.s.}{=} \bm{f}(\textbf{Y}^{\A})$ is that $\ppsi_k(\bm{x}_{j}) =\ppsi_k\ar(\A_{j})$ for all $j$ and $k$. 
		\item \textbf{null-MAUP-theorem:} A necessary and sufficient condition for $\bm{f}(\textbf{Y}^{\B}) \overset{a.s.}{=} \bm{f}(\textbf{Y}^{\A})$ is that $\ppsi_k\ar(\B_{j}) =\ppsi_k\ar(\A_{j})$ for all $j$ and $k$. 
	\end{enumerate}
\end{proposition}
\noindent
\textbf{Proof:} See Appendix A.\\

\noindent
The null-MAUP-theorem essentially states that if the lower resolution eigenfunctions are piece-wise constant over the higher resolution $\A_{1},\ldots, \A_{n_{A}}$ then the statistic $f(\textbf{Y}^{\B})$ at the lower resolution is almost surely no different from the statistic at the higher resolution $f(\textbf{Y}^{\A})$, suggesting that there is no between-scale-differences in the conclusions based a generic statistic, and hence no MAUP (the null-ecological-fallacy theorem is interpreted in a similar way). This interpretation of the null-MAUP-theorem highlights why the MAUP and ecological fallacy occur; namely multiscale error occurs the eigenfunctions express functional variability across scales. Consequently, a natural metric to determine the presence of the MAUP/ecological fallacy is the between scale variance of the eigenfunctions; that is,
\begin{align} \label{eq:truemvcage}
	\mathbb{V}_{j}(\A,\{\lambda_{k}\},\{\ppsi_k (\cdot)\}) \equiv  \int \sum_{k = 1}^{\infty}\frac{\lambda_k (\psi_{jk}(\A) - \psi_{jk}(\s))^{2} }{\hspace{2pt}|A|}d\s.
\end{align}
which is estimated with what we call MVCAGE:
\begin{align} \label{eq:estmvcage}
		\mathrm{MVCAGE}(\A) =  \mathbb{E}\left[\sum_{j = 1}^{N}\mathbb{V}_{j}(\A,\{\lambda_{k}\},\{\ppsi_k (\cdot)\})\hspace{2pt}\big\vert \hspace{2pt}\bm{z}_{1},\ldots, \bm{z}_{N}\right],
\end{align}
where the $n$-dimensional vectors $\bm{z}_{j} = (Z_{j}(\s_{1}),\ldots, Z_{j}(\s_{n}))^{\top}$, and $\mathbb{E}$ is the expected value operator. In (\ref{eq:estmvcage}), MVCAGE is taken with respect to a posterior distribution, where we have implicitly assumed a data model, process model, and prior distributions (see Section~\ref{sec:model} for more details). There is an alternative way to express MVCAGE in terms of covariance matrices, and we provide these details in Appendix B. Additionally, the logic leading to MVCAGE via the null-MAUP-theorem does not require the use of the squared error loss, and one can easily swap in a different loss function (see Appendix C for a brief discussion). 

The univariate CAGE from \citet{bradley2017regionalization} can be written as,
\begin{equation}
	\mathrm{CAGE}_{j}(\A) = \mathbb{E}\big[\mathbb{V}_{j}(\A,\{\lambda_{k}\},\{\ppsi_k (\cdot)\})\hspace{2pt}\big\vert \hspace{2pt} \bm{z}_{j}\big]; \hspace{2pt}j = 1,\ldots, N,
\end{equation}
 is equivalent to MVCAGE in the case when $N = 1$. There is a clear benefit to incorporating/leveraging multivariate dependence. That is, suppose we are given a data model for $Z_{j}\vert Y_{j}, \bm{\theta}_{D}$, process model for $Y_{j}\vert \{\lambda_{k}\}, \{\ppsi_{k}(\cdot)\},\bm{\theta}_{P}$, and prior distributions for real-valued data specific parameters $\bm{\theta}_{D}$ and process parameters $\{\lambda_{k}\}$ and $\bm{\theta}_{P}$. Then we have the following property,
 \begin{equation}\label{eq:mvcagebetter}
 \mathbb{E}\left[\left\lbrace\sum_{j = 1}^{N}\mathbb{V}_{j}(\A,\{\lambda_{k}\},\{\ppsi_k (\cdot)\}) - \mathrm{MVCAGE}(\A)\right\rbrace^{2} \right]\le \mathbb{E}\left[\left\lbrace\sum_{j = 1}^{N}\mathbb{V}_{j}(\A,\{\lambda_{k}\},\{\ppsi_k (\cdot)\}) - \sum_{j = 1}^{N}\mathrm{CAGE}_{j}(\A)\right\rbrace^{2}\right],
 \end{equation}
 \noindent
where the expectation is taken with respect to the joint distribution of $\{\lambda_{k}\}$, $\bm{z}_{1},\ldots,\bm{z}_{N}$. This statement is easily proven, since it is well known that the posterior mean minimizes the squared error among all real-valued functions of the entire dataset $\bm{z}_{1},\ldots, \bm{z}_{N}$ \citep[e.g., see][among others]{berger2013statistical}. Equation (\ref{eq:mvcagebetter}) follows immediately. This gives the simple yet important statement that we are better able to estimate $\sum_{j = 1}^{N}\mathcal{V}_{j}(\A,\{\lambda_{k}\},\{\ppsi_{k}(\cdot)\})$ by leveraging multivariate spatial dependence than when leveraging spatial-only dependence. 
}

\subsection{The Multivariate OC Basis Function} \label{klecmp}

{The use of the truncated KLE can be restrictive as the basis functions are required to be orthogonal. However, there are a wide range of options for non-orthogonal spatial basis functions. OC basis functions provides a solution to this issue, and allow one to re-weight non-orthogonal basis functions to imply orthogonality that satisfies the Fredholm-integral equation. However, OC basis functions have primarily been used in spatial/functional settings \citep{obled1986some}. The OC basis function vector for the $j$-th spatial process can be written as,
\begin{equation*}
	\big(\psi_{j1}^{OC}(\s), \ldots, \psi^{OC}_{j\widetilde{M}_j} (\s)\big) = \big(\theta_{j1}(\s), \ldots, \theta_{j\widetilde{M}_j} (\s)\big)\textbf{F}_{j},
\end{equation*}	
\noindent 
where $\textbf{F}_{j}$ is a $\widetilde{M}_{j}\times \widetilde{M}_{j}$ real-valued matrix and let $\{\theta_{jk}\}_{k=1}^{\widetilde{M}_j}$ be non-orthogonal basis functions called generating basis functions (GBF). This leads to the following result that shows that one can extend OC basis functions to the multivariate spatial settings.

\begin{proposition} \label{prop33}
Define the matrix $\bm{W}_m$ with $b, \ell$-th element given by $w^m_{b\ell} = \int_{\mathcal{S}} \theta_{mb} (\s) \theta_{m\ell} (\s) \, d\s$ for $m = 1,\ldots, N$. Assuming positive-definiteness of $\bm{W}_m$, and let $\mbQ_m$ be defined as the inverse Cholesky decomposition of $\bm{W}_n$, i.e., $\bm{W}_m^{-1} = \bm{Q}_m\bm{Q}_m^{\top}$. Assume, $\textbf{F}_{m} = \textbf{Q}_{m}$, and define the $j$-th element of the $k$-th multivariate OC eigenvector to be $\big[ \ppsi_{k}^{OC} \big]_j = \big(\psi_{j1}^{OC}(\s), \ldots, \psi^{OC}_{j\widetilde{M}_j} (\s)\big)\bm{e}_k^j$, where $\bm{e}_k^j$ is defined above Equation~\ref{cons1}. Then $\ppsi_{k}^{OC}$ for $k = 1,\ldots,$ admits a multivariate KLE.
\end{proposition}
\begin{proof}
	See Appendix A.
\end{proof}
\noindent
From Proposition~3.3, the construction of the multivariate OC (MVOC) basis functions essentially repeats the spatial-only OC construction nested within the multivariate KLE construction from \citet{happ2018multivariate}. What is particularly interesting is that we immediately obtain \citet{happ2018multivariate}'s construction when $\theta$ is orthogonal, in which case $\textbf{F}_{m}$ is equal to the identity matrix. From this perspective \citet{happ2018multivariate}'s construction can be seen as a special case of MVOC, as MVOC allows the generating basis functions $\theta_{jk}$ to be non-orthogonal, whereas \citet{happ2018multivariate} requires orthogonal bases.
}

\section{Statistical Model and Implementation} \label{sec3}
{Here we describe our model and implementation based on our novel multiscale multivariate KLE. We first discuss a key pre-processing steps in Section~\ref{sec:discrete}). Then we present the Bayesian model (Section~\ref{sec:model}) and estimation of the MKLE in Section~\ref{sec:eof}.} Finally, Section~\ref{sreg} proposes the complete algorithm for spatial regionalization.

\subsection{Discrete Approximations}\label{sec:discrete}
In practice it is not possible to integrate the eigenfunctions $\ppsi_k$ over all possible partitions since there can be an infinite number of partitions of the space. One way to proceed is to consider a Monte-Carlo approximation of the space using a grid of points \citep[e.g.,][]{daw2022supervised}. The idea is similar to methods used in Riemann integration, decision trees, and random forests. We use a set of pseudo points $\widetilde{\s}_1, \ldots, \widetilde{\s}_{\widetilde{n}} \in \SSS$ that lie at the centroid of regular grid ``boxes''. Suppose, $\{ \B_j\}_{j=1}^{\widetilde{n}}$ is such a set of rectangular grid boxes covering $\SSS$. One can use Monte-Carlo integration then to evaluate the areal eigenfunctions $\ppsi\ar(\B_j)$ over these grid boxes. Alternatively, similar to approximations used in Riemann integration, we can use sufficiently small grid boxes and assume that the MKLEs do not vary substantially inside these boxes. We can then approximate the areal eigenfunctions in the following manner:
\begin{align} \label{grideq} 
	\ppsi^{\ar}(\B_j) = \frac{1}{|\B_j|} \int_{\B_j} \ppsi(\s) \, d\s \approx \frac{1}{|\B_j|}\ppsi(\widetilde{\s}_j) \int_{\B_j}  \, d\s = \ppsi(\widetilde{\s}_j).
\end{align} 
where $\widetilde{\s}_j$ is a point from $\B_j$. We can join the grid elements in \eqref{grideq} to form the lower-resolution areal units, i.e., $\A_k = \cup_{j=1}^{n_k} \B_{k_j}$. Similar to the above derivation, the areal eigenfunction at the aggregated areal level $\A_k$ becomes the average of the eigenfunctions over the grid elements, i.e., $\ppsi^A(\A_k) = \frac{1}{n_k}\sum_{j=1}^{n_k} \ppsi^A(\B_{k_j})$. Under this approximation, the discrete approximation of MVCAGE over $\A_k$ becomes 
\begin{align} \label{dmv}
	\mathrm{DMVCAGE}(\A_k) = \frac{1}{n_k} \sum_{j=1}^{n_k} \sum_{i=1}^{\infty}  {E\left[\lambda_k \big[\ppsi_i\ar(\A_k) -  \ppsi_i\ar(\B_{k_j}) \big]^{\top} \big[\ppsi_i\ar(\A_k) -  \ppsi_i\ar(\B_{k_j}) \big]\vert \Z_{1},\ldots, \Z_{N}\right]},
\end{align}
{
	\noindent
	where the second sum stops at $M$ in practice.
}

\subsection{The Statistical Model} \label{sec:model}
Recall that $\Z(\s)$ and $\Y(\s)$ are the observed variables and underlying processes, respectively. We also denote the vector-valued mean function and error process as $\bm{\mu}(\s)$ and $\bm{\epsilon}(\s)$, respectively. For example, consider the following model
\begin{align} \label{ocmod}
	\Z(\s) &= \bm{\mu}(\s) + \Y(\s) + \bm{\epsilon}(\s),
\end{align}
where $\bm{\mu}(\s) = \bm{\mu}$ and $\bm{\epsilon}(\s) \overset{iid}{\sim} \mathcal{N}\big(\bm{0}, \diag (\sigma_1^2, \ldots, \sigma^2_N)  \big)$. Models of this form are common in many multivariate spatial applications \citep[e.g., see][]{banerjee2004hierarchical, cressie2015statistics} so the difference in models then becomes how one parameterizes the dependence structure of $\Y(\s)$. Many researchers use full-rank cross-covariance models to estimate the joint covariance matrix \citep[e.g., see the review in ][]{genton2015cross}. Alternatively, one can use a low-rank estimate of the covariance matrix using a basis function model, which was the primary idea in univariate CAGE \citep{bradley2017regionalization}. We extend this to multivariate processes and call it the multivariate Obled-Cruetin (MVOC) approach as explained below.

We consider the following modeling assumptions in our example to demonstrate the MVOC approach. For the $j$-th univariate process, we select a set of prescribed orthonormal basis functions: $\{\phi_{jk}(\s): \SSS \to \mathbb{R} \}_{k=1}^{\widetilde{M}_j}$. Using these basis functions, we consider the following modeling assumption  
\begin{align} \label{jbas}
	Y_j(\s) &= \sum_{k=1}^{\widetilde{M}_j} \phi_{jk}(\s) \nu_{jk}. 
\end{align}
\noindent We denote the orthonormal basis matrix $\bm{\Phi}_j \in \mathbb{R}^{n \times \widetilde{M}_j}$ as the $n \times \widetilde{M}_j$ matrix with $k, \ell$-th element as $\phi_{j\ell}(\s_k)$. These are the univariate Obled-Creutin (OC) bases that were used in the derivation of the univariate CAGE criterion in \citet{bradley2017regionalization}. We then define the joint OC basis as the block matrix $\bm{\varPhi}$ with $i,j$-th block as $\bm{\Phi}_j$ for $i = j$, $\bm{0}$ when $i \neq j$. Next, denote $\Z\supvec$ as the $Nn \times 1$ vector version of the $n\times N$ matrix of the observations $\Z(\s)$. Define $\Y\supvec$, $\bm{\nu}\supvec$, and $\bm{\epsilon}\supvec$ similarly. Then, in matrix form, the modeling assumptions \eqref{ocmod} and \eqref{jbas} are equivalent to 
\begin{align*}
	\Z\supvec &= \bm{\mu}\supvec + \Y\supvec+ \bm{\epsilon}\supvec, 
\end{align*}
where $\Y\supvec = \bm{\varPhi} \bm{\nu}\supvec$. Now, within a Bayesian inferential framework, we use the following prior assumptions
\begin{align*}
	\mu_j &\propto 1, \\
	\epsilon_j(\s) &\overset{iid}{\sim} \mathcal{N}(0, \sigma^2_j), \\
	\bm{\nu}_j = \big( \nu_{j1}, \ldots, \nu_{jM_j} \big) &\sim \mathcal{N} (0, g \sigma_j^2 \big(\bm{\Phi}_j^{\top}\bm{\Phi}_j\big)^{-1}), \\
	\sigma_j^2 &\sim \operatorname{Inverse} \operatorname{Gamma}(a^{\sigma}_j, b^{\sigma}_j), 
\end{align*}
\noindent where in our examples, we used $a^{\sigma}_j$ = $b^{\sigma}_j$ = $0$, which corresponds to the Jeffrey's prior $\sigma_j^2 \propto \frac{1}{\sigma^2}$. The above formulation is known as Zellner's g-prior \citep{zellner1986introduction}, which is often used in penalized regression and variable selection \citep[e.g., see a review in][]{li2021zellner}. After experimenting with many possible choices, we use $g = n$ in our examples, which is  interpreted as the prior having the equivalent weight of one observation. Note that more complicated models and prior choices can be used here, and our only goal is to estimate the joint covariance matrix $\CCC$. The posterior samples are obtained via a block Gibbs sampler, where each parameter is sampled in sequence from its full conditional distribution (refer to Appendix \ref{computationgibbs} for the details). 

\subsection{Estimating the Multivariate KLE via Empirical Orthogonal Functions}\label{sec:eof}

{We take the strategy of estimating the multivariate KLE using estimates of the covariance matrix. Note that in spatial statistics, eigenvectors computed from empirical covariance matrix estimates are referred to as empirical orthogonal functions \citep[EOFs, e.g.,][]{cressie2015statistics}}. {Here,} for cross-covariance models, we employ the linear model for co-regionalization (LMC) \citep{journel1978mining, goulard1992linear} to estimate the full covariance matrix of {$\bm{\Z}\supvec \equiv (\bm{\Z}_{1}^{\top},\ldots, \bm{\Z}_{N}^{\top})^{\top}$}.  Interested readers can see \citet{wackernagel1989overview} for a survey on co-regionalization.  Specifically, we use the co-kriging procedure from the \texttt{R} package \texttt{gstat} \citep{pebesma2015package, gstat2, rossiter2007co}. For a bivariate process, the LMC first models one of the spatial processes (say $Z_1(\s)$) using the traditional geospatial model \citep{cressie2015statistics}. Then, the second spatial process is built as a conditional model using $Z_1(\s)$, i.e., we estimate the distribution of $\pi \big[Z_2(\s) | Z_1(\s) \big]$. We use this approach to estimate the underlying parameters of the covariance matrix (cf. Section~\ref{oceanex} for an example).

To estimate the KLE from the MVOC basis model, we use the following procedure. From the posterior {Markov chain Monte Carlo (e.g., from LMC)}, {estimate the posterior covariance matrix of the multivariate spatial random process and denote it with} $\widehat{\bm{\Sigma}}$, and perform the eigendecomposition as {$\widehat{\bm{\Sigma}} = \widehat{\bm{E}} \widehat{\bm{\Lambda}} \widehat{\bm{E}}^{\top}$}. Similar to Section~\ref{sec:smkle}, denote the $k$-th column of {$\widehat{\bm{E}}$} as the vector {$\widehat{\bm{e}}_k$} and express it in the following block-vector form: {$\widehat{\bm{e}}_k^{\top} = \big(\widehat{\bm{e}}_k^{1}, \ldots, \widehat{\bm{e}}^{N}_k \big)$}. Then, {to compute MVCAGE we use the diagonal elements of $\widehat{\bm{\Lambda}}$ to compute the eigenvalues, and} the $j$-th element of the $k$-th eigenfunction and the expansion coefficients are given by 
{
	\begin{align} \label{cons12}
		\big[ \ppsi_{k}(\s) \big]_j &= \big( \phi_{j1}(\s), \, \ldots, \,  \phi_{jM_j}(\s) \big) \widehat{\bm{e}}_k^j 
	\end{align}
	\noindent
	which follows from \citep{happ2018multivariate}. {We substitute $\psi_{jk}^{OC}$ for $\phi_{jk}$ in (\ref{cons12}).}} Note that $M_j$ is always less than or equal to $\widetilde{M}_j$, and one can set $M_j = \widetilde{M}_j$. Alternatively, one can discard the trailing eigenvalues of $\mathbf{U}_j$ if their values are too small and indistinguishable from each other \citep[e.g., ][]{hastie2009elements}. {There are several alternative modeling to estimate the MKLE. We provide this discussion in Appendix~\ref{klecmp2}.}

\subsection{Spatial Regionalization} \label{sreg}
{Regionalization stats with a collection of $J$ regionalizations, $\{\A_{1}^{(1)},\ldots, \A_{n_{1}}^{(1)}\},\ldots, \{\A_{1}^{(J)},\ldots, \A_{n_{1}}^{(J)}\}$. The optimal regionalization is given by,
\begin{equation}
	\{\A_{1}^{*},\ldots, \A_{n_{1}}^{*}\} = \underset{j = 1,\ldots, J}{\mathrm{arg\hspace{2pt}min}}\left\lbrace \sum_{i = 1}^{n_{j}}\mathrm{MVCAGE}(\A_{i}^{(j)}) \right\rbrace,
\end{equation}
which we can compute} given the multivariate eigenfunctions and eigenvalues from Sections~\ref{sec:eof} (or Appendix \ref{klecmp2}). Here, we apply a similar two-stage regionalization algorithm as in univariate CAGE \citep{bradley2017regionalization}. In the first stage, we employ a spatial clustering algorithm (e.g., spatial $k$-means, HGC, MST) that finds a set of {$J$} candidate supports. In the second stage, the MVCAGE (or DMVCAGE) statistic is computed over these supports and the minimum statistic is retained. \citet{bradley2017regionalization} argues that in the context of the univariate CAGE, this methodology is more computationally efficient than a global search of minimum CAGE since the total number of possible choices is of combinatorial order in the number of partitions. This is also true for MVCAGE. We used the HGC with Ward clustering on the space of the KLE, which generated spatially contiguous partitions. An additional benefit of using HGC is that, since we need to compare MVCAGE statistic for different choices of number of clusters, we can implement parallel processing after computing the linkages (or dendrogram) once. {For additional discussion on clustering algorithms see Appendix~\ref{app:cluster}.}

Unlike the univariate CAGE, where the authors used two hyperparameters for the maximum and minimum number of areal units considered, we use one user-defined stopping criterion in our study to find an optimal choice of number of clusters. For a given number of clusters, we compute the MVCAGE statistic and keep proceeding until the relative change in MVCAGE for two consecutive clusters becomes very close. Although this choice of the stopping parameter is user-defined, one can use similar ideas as the elbow or Silhouette plot for $k$-means clustering to have a data-driven guess. The complete algorithm for regionalization with MVCAGE is given by Algorithm \ref{mvcagealg}.
\begin{algorithm}
\caption{Spatial Regionalization with MVCAGE}
\begin{algorithmic}[1]
\algblock{Input}{EndInput}
\algnotext{EndInput}
\algblock{Output}{EndOutput}
\algnotext{EndOutput}
\newcommand{\Desc}[2]{\State \makebox[2em][l]{#1}#2}
\Input \, Multivariate data $\Z(\s) = \big(Z_1(\s), \ldots, Z_N(\s) \big)^{\top}$ at locations $\bm{S} =\{\s_1, \ldots, \s_n \}$. Model specification $\mathbb{M}$. Stopping parameter $\epsilon$ for the number of regions under consideration.
\EndInput
\State Fit the model $\mathbb{M}$ to estimate the underlying processes $\Y(\s)$ and their covariance functions (Section~\ref{sec:model}).
\State {Compute the empirical orthogonal functions (Section~\ref{sec:eof} or Appendix~\ref{klecmp2}). }
\State $j=1$.
\State \textbf{\textit{while}} true \textbf{\textit{do}}

a: Use HGC on $\Y(\s_1), \ldots, \Y(\s_n)$ to find $j$ number of candidate supports.

b: Compute the MC$_j$ = MVCAGE for these candidate supports.

c: \textbf{\textit{if}} $\frac{\operatorname{MC}_{j-1} - \operatorname{MC}_j}{\operatorname{MC}_{j-1}} < \epsilon$ \textbf{\textit{break}}.

d: $j = j +1$.

\State \textbf{\textit{end while}}

\State Return the clustering output of HGC.
\end{algorithmic}
\label{mvcagealg}
\end{algorithm}

\section{Simulation and Applications} \label{sec4}
We illustrate our methodology in this section using simulated data and data from demographic and ocean color applications. We use a simulated dataset with repeated observations in Section~\ref{mvfda} and demonstrate the MVCAGE regionalization using empirical KLE. In Section~\ref{ocex}, using American Community Survey (ACS) and hospital quality data, we show the application with MVOC basis functions and Bayesian computation of the MKLE.  Then in Section~\ref{oceanex}, we demonstrate the covariance estimation using LMC and apply it to an ocean color dataset. 

\subsection{Simulation: Empirical KLE-based MVCAGE for Functional data} \label{mvfda}
We consider one-dimensional (in space) bivariate simulated data from a full bivariate Mat\'ern process \citep{guttorp2006studies, gneiting2010matern} over the one-dimensional region $[0,1]$ with the covariance functions $\mathbb{C}$ that take the following form:
\begin{align*}
    \mathcal{M}(\s, \bm{r}; \nu, a) &= \frac{2^{1- \nu}}{\Gamma(\nu)} \big( a \|\s - \bm{r}\| \big)^{\nu} \mathbb{K}_{\nu}(a\|\s - \bm{r}\|), \\
    \CC_{ii}(\s, \bm{r}) &= \sigma_i^2 \mathcal{M}(\s, \bm{r}; \nu_{ii}, a_{ii}), \\
    \CC_{ij}(\s, \bm{r}) &= \rho \sigma_i \sigma_j \mathcal{M}(\s, \bm{r}; \nu_{ij}, a_{ij}). 
\end{align*}
\noindent We use the following parameter set for the simulation: $a_1 = 10$, $a_2 = 15$, $a_{12} = 1.2 \max(a_1, a_2)$, $\nu_1 = 0.4$, $\nu_2 = 0.5$, $\nu_{12} = 0.5 (\nu_1 + \nu_2)$. We randomly select $1000$ locations over $[0,1]$ and generate $r = 4000$ replications of the bivariate spatial data. 

Since we have multiple replications here, we use the low-rank empirical MKLE (Equation \ref{empmkle}) in this example. We compute the MKLE of the joint covariance kernel using $K = 50$ sets of Fourier basis functions, i.e., $\phi_{j1}(s) = 1$, $\phi_{jk}(s) = \big(\sin(2 \pi \frac{k-1}{50}s), \cos(2 \pi \frac{k-1}{50}s)\big)$ for $k = 1, \ldots, 50$. This choice leads to a total of $101$ basis functions. We use these basis functions to estimate the expansion coefficient and then decorrelate them to compute the MKLE of the covariance kernels. Based on that, we compute the MCAGE loss function.

We next use $5,000$ pseudo-points over the domain for the clustering. We consider $10^{-4}$ as the stopping crierion. The final selection of the number of clusters is $97$. Figure~\ref{samp_fdata_cl} shows a realization of the simulated data (i.e., $Z_1(s)$ and $Z_2(s)$) and demonstrates its aggregated version (i.e., $\widehat{\Z}(s)$) after applying MVCAGE. Figure~\ref{samp_fdata_cl2} shows the quantity of MVCAGE.
\begin{figure}[ht!]
    \centering
    \includegraphics[height=6cm, width=14cm]{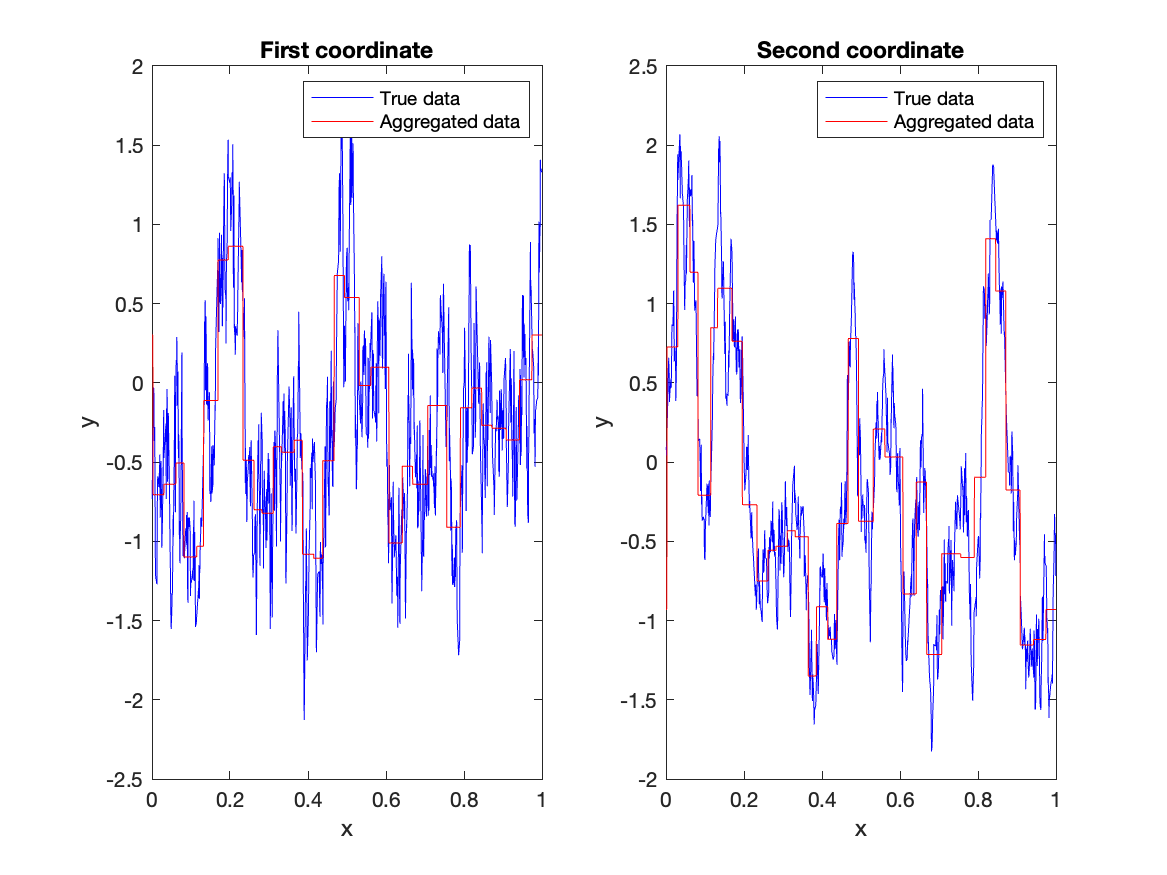}
    \caption[Regionalization with MVCAGE for Data from Bivariate Mat\'ern Process]{{\bf Regionalization with MVCAGE for data from a one-dimensional bivariate Mat\'ern process: } We applied MVCAGE to aggregate sample data from bivariate Mat\'ern Process. The blue lines correspond to the simulated data\footnote{\baselineskip=10pt \color{red} Ranadeep: ``True data" sounds awkward and is not consistent with the legend. I think that we should change this to ``Simulated data".} over $1000$ locations between [0,1]. The left panel shows the original (in blue) and aggregated data (in red) for the first process. The right panel shows the same for the second process.}
    \label{samp_fdata_cl}
\end{figure}
\begin{figure}[ht!]
    \centering
    \includegraphics[height=8cm, width=14cm]{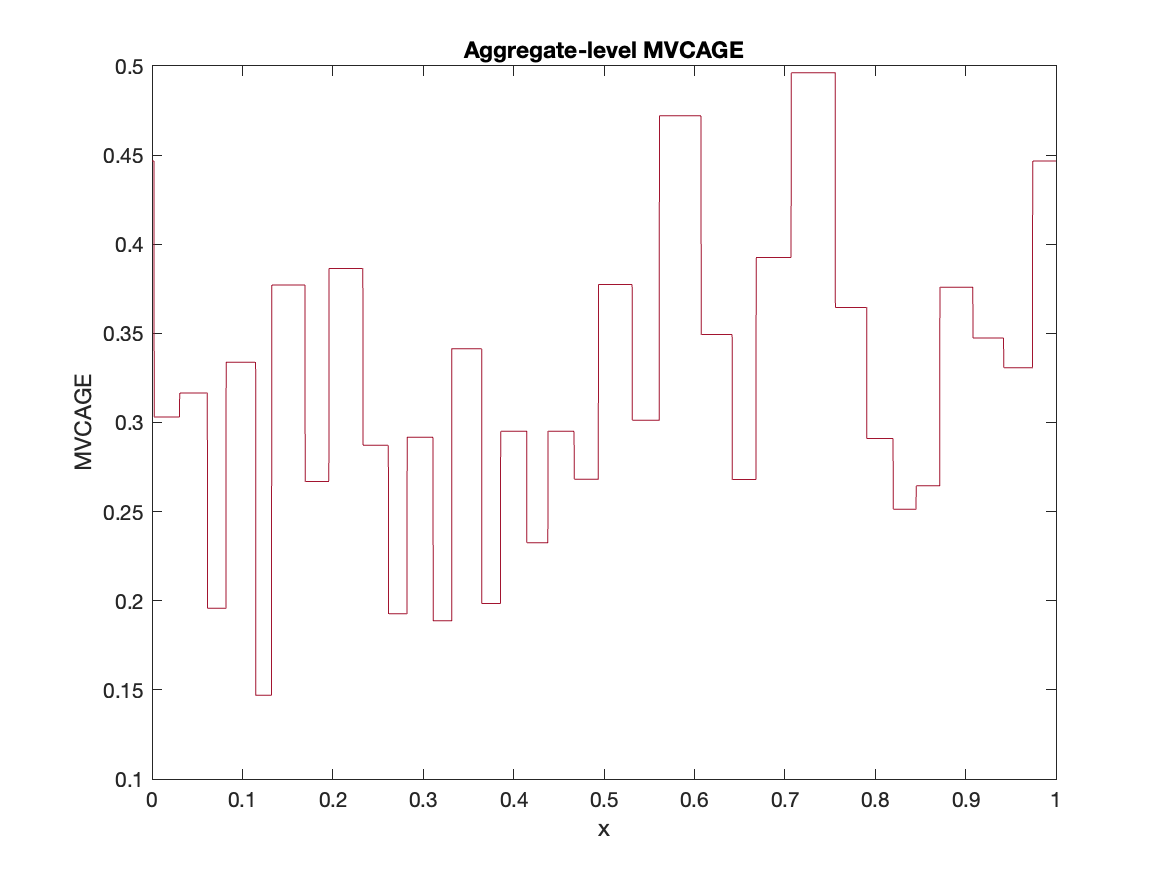}
    \caption[Errors for Aggregation of sample functional data]{ {\bf Value of MVCAGE over the areal units:} Plot of the MVCAGE statistic over the areal regions for the optimal choice of regionalization. The larger the MVCAGE value, the greater is the potential for the ecological fallacy to occur during inference.}
    \label{samp_fdata_cl2}
\end{figure}

\subsection{U.S. County-level Regionalization} \label{ocex}
We apply our methodology to regionalize the counties of the United States using a bivariate dataset consiting of log median income and hospital quality ratings. We considered the data from $3106$ counties in $2015$ that had no missing data. The log median income data consists of public-use county-level 5-year period estimates from the American Community Survey\footnote{\url{https://www.census.gov/programs-surveys/acs}} and can be easily accessed using $\texttt{R}$ package $\texttt{tidycensus}$. The hospital quality ratings are obtained through the Dartmouth Atlas Study\footnote{ \href{https://www.dartmouthatlas.org/}{https://www.dartmouthatlas.org/}}, which is a research initiative focused on healthcare analysis across the United States. We use the adjusted ratings data for the U.S. counties based on the primary care access and quality measures available at the respective hospitals. 

We used a MVOC basis approach here. To construct the spatial OC basis, we started with $300$ Gaussian basis functions as the GBF. We computed $300$ OC basis functions from these GBFs. For the regionalization algorithm, we used a stopping criterion of $0.01$ and the final number of areal units is $96$. The final constructed areal data is shown in Figure~\ref{maxminimg}. The regionalization is shown in Figure~\ref{maxminimg2}. The MVCAGE regionalization shows 97 broad contiguous spatial regions in both variables. This represents a substantial dimension reduction from the county data.

These two particular variables allow us to assess the discrepancy in areas in need (i.e., log income can be used as a proxy for an area in need) and the quality of the hospitals in that area. In the bottom-left panel of Figure~\ref{maxminimg}, we see that relatively lower log incomes tend to arise in the southeast, southwest, and northwest US, with central areas and northeast regions of the US having slightly larger log incomes. In contrast, in the bottom-right panel, the hospital quality metrics tend to be larger only in eastern US. This suggests that several regions on the West Coast and central US have an indirect relationship between log income and hospital quality. Regionalization is particularly useful for coming to these conclusions as the fine-scale features of the country-level observations obfuscate these trends.

We note that this analysis is meant as an illustration and did not directly consider the sampling error variance in the ACS survey data. In principle, such uncertainties could be included in the data model of the hierarchical framework, possibly leading to different regionalizations. The consideration of measurement and sampling error on regionalization is an interesting topic, but is beyond the scope of the current work.

\begin{figure}[ht!]
    \centering
    \begin{subfigure}[t]{0.5\textwidth}
        \centering
    \includegraphics[height=18cm, width=9cm]{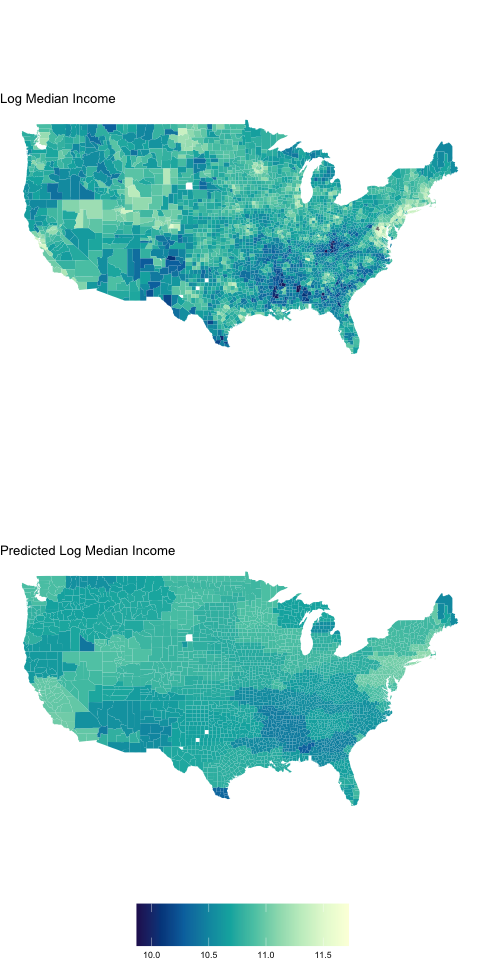}
\end{subfigure}%
    ~ 
    \begin{subfigure}[t]{0.5\textwidth}
        \centering
    \includegraphics[height=18cm, width=9cm]{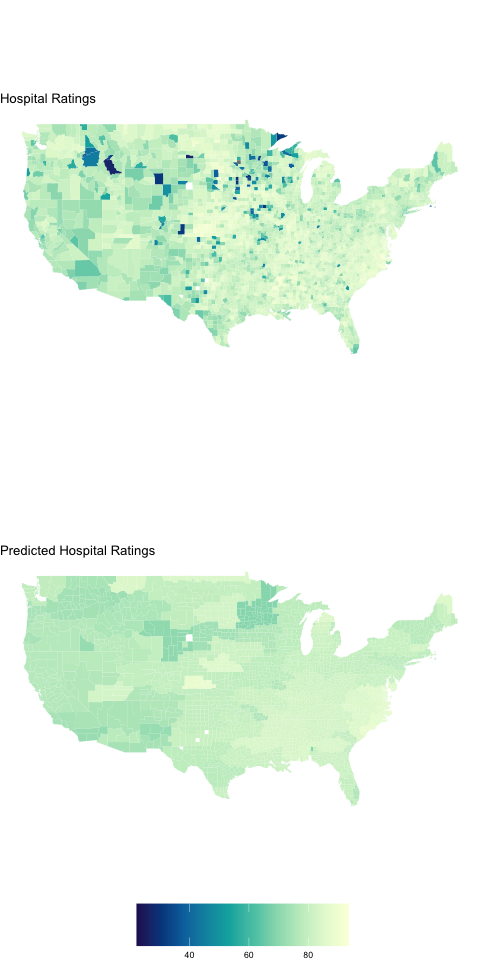}
\end{subfigure}
    \caption[Regionalization of Bivariate log median income and Hospital Quality Data over U.S. counties with MVCAGE]{{\bf Regionalization of Bivariate log median income and Hospital Quality Data over US counties with MVCAGE}: We applied our methodology to aggregate the counties over the U.S. using the ACS 5-year period estimate of log median income and the Dartmouth Atlas Study hospital quality ratings from 2015. The left column shows the county-level (top) and the area-level (bottom) aggregation of the log median income data. The right column demonstrates the same for the quality ratings.}
    \label{maxminimg}
\end{figure}

\begin{figure}[ht!]
    \centering
    \includegraphics[height=10cm, width=10cm]{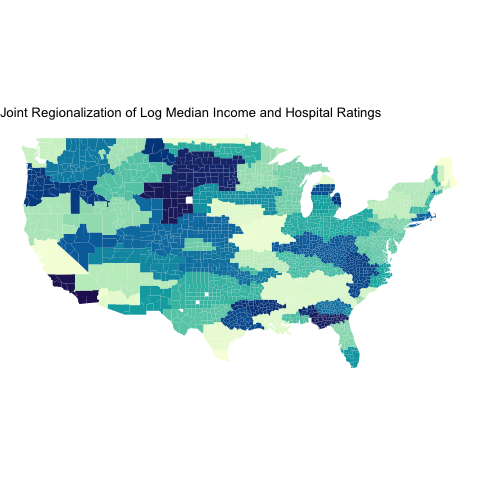}
\caption[Regionalization of Bivariate log median income and Hospital Quality Data over U.S. counties with MVCAGE: ]{{\bf Regions using MVCAGE for Bivariate Regionalization of U.S. counties with MVCAGE}: We show the regions from the bivariate regionalization of the U.S. counties the ACS 5-year period estimate of log median income and the Dartmouth Atlas Study hospital quality ratings from 2015.}
    \label{maxminimg2}
\end{figure}

\subsection{Prediction of Ocean Color} \label{oceanex}
We apply our methodology to perform a bivariate aggregation in the coastal Gulf of Alaska \citep{leeds2014emulator, wikle2013modern} based on the satellite data of ocean color and ROMS ocean model  (Regional Ocean Modeling System) \citep{doi:10.1175/2010MWR3421.1} output for the amount of chlorophyll. There are many satellites focused on the collection of ocean color datasets, e.g., SeaWiFS, MODIS, MERIS, etc. Ocean color provides information on the quantity of phytoplankton { in the water column near the ocean surface.} Therefore, it leads to the inference of the primary productivity and ecology of the upper levels of the ocean. Moreover, the measurement of chlorophyll is another important covariate correlated with species distribution and ocean ecology through its fundamental role as a food source in the lower levels of the food chain. ROMS is a three-dimensional (in space) ocean circulation model that simulates oceanic and estuarine processes at regional scale over time. The model can be used to simulate water temperature, salinity, currents, and other physical and biogeochemical variables, including chlorophyll content. Therefore, we consider a spatial aggregation or regionalization based on the joint distribution of ocean color from the SeaWiFS satellite and ROMS ocean model output for chlorophyll. { Obtaining a common reduced-dimension set of areal units can be used to facilitate data fusion as in \citet{leeds2014emulator}.
}

We consider the data from 12 May 2000, which has complete observations at $4718$ spatial coordinates. We use the LMC method from Appendix~\ref{klecmp2} here to estimate the joint covariance matrix with the bivariate Mat\'ern kernel. Then, we applied our two-stage regionalization algorithm to minimize the value of the MVCAGE statistic. Figure~\ref{ocean} shows the results of using our approach on these data. We choose the lower and upper bounds to be $250$ and $350$ and note that the regionalization with $306$ areal regions had the minimum amount of MVCAGE. Hence, this is an order-of-magnitude reduced representation of the original data in the aggregated scale such that the amount of ecological fallacy in the aggregation is minimized. 

\begin{figure}[ht!]
    \centering
    \includegraphics[height=10cm, width=16cm]{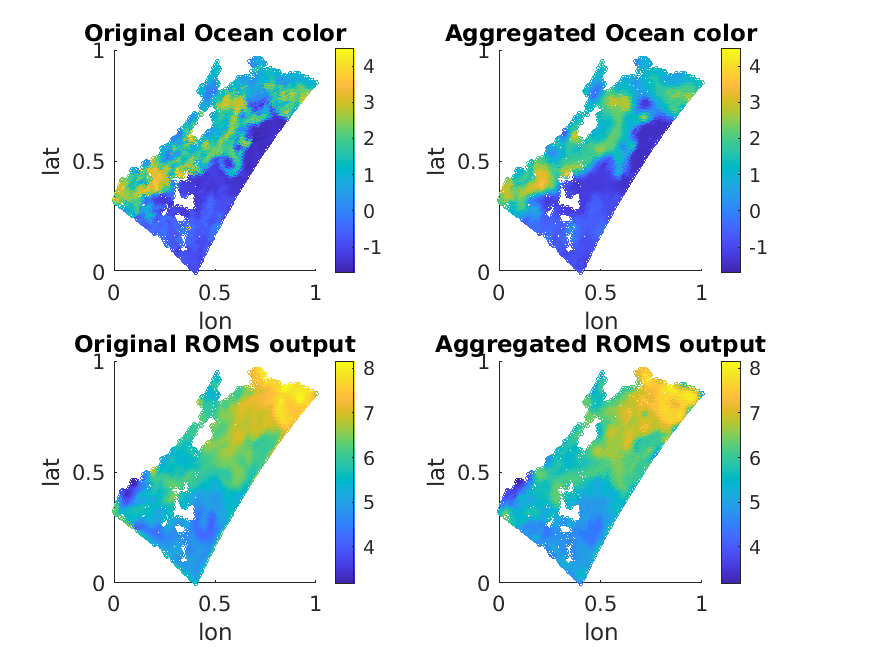}
    \caption[Regionalization of Ocean Color and ROMS Model Output]{{\bf Regionalization of Ocean Color and ROMS Model Output}: We applied our methodology to aggregate the SeaWiFS ocean color data which is an indicator of the presence of chlorophyll. The top row shows the point level (left) and the areal level (right) aggregation of the observed ocean color. The bottom row demonstrates the same for the ROMS  ocean model output for the amount of chlorophyll.}
    \label{ocean}
\end{figure}

%
\section{Discussion} \label{sec5}
In this manuscript, we {we provide theoretical justification for a multivariate CAGE (MVCAGE), which includes multiscale multivariate extensions of the KLE, Mercer's theorem, and the-null-MAUP theorem. Additionally, we provide technical results that demonstrate that the MVCAGE is preferable to spatial-only CAGE, and extend commonly used basis functions used to compute CAGE to the multivariate spatial setting. Motivated by this theoretical development we proposed a model that can be used for} multivariate regionalization of spatial data. Here we demonstrate how to build {multiscale} multivariate eigenfunctions starting from the univariate process-specific and cross-covariance matrices. Moreover, we also {extend} the Obled-Cruetin basis functions to {the multiscale multivariate spatial setting}. All the procedures are similar in the following sense. We start with any set of orthonormal basis functions and assume that the multivariate eigenfunctions lie on the span of these basis functions. We then compute the expansion coefficients and their joint covariance matrix. Decorrelating this matrix gives us the necessary eigenpairs.

The principal goal of applying the {multiscale} MKLE-based loss function is to minimize the ecological fallacy in a spatial COS procedure. Spatial COS changes the resolution of the spatial supports. It often comes with Simpson's paradox-like behavior in that the original and the aggregated data may demonstrate two different spatial patterns. Although researchers have studied the ecological fallacy extensively, there has been little development to quantify the amount of the ecological fallacy and minimize this when performing spatial aggregation. Here, we use the logic of the univariate CAGE method by \citet{bradley2017regionalization} and extend it to the multivariate domain. We demonstrate how to compute the area-level eigenfunctions, defined as the continuous average of the corresponding point-level eigenfunctions over any given areal region. We then use any loss function (such as $\mathcal{L}_1$, $\mathcal{L}_2$) to compute the distance between the point-level and areal-level eigenfunctions, with eigenvalues used as weights. The MVCAGE loss over any areal region is the continuous average of this loss function. We provide a discussion and proposition on why minimizing the MVCAGE loss leads to a lower value of the ecological fallacy in spatial aggregation.

In practice, we need to use a discrete approximation of the MVCAGE statistic using a set of gridded points and Monte Carlo approximations. If the original data is in an areal format, we can approximate the eigenfunctions within each areal region by using the Monte Carlo approximation. If the data is in a continuous point format, we can use regular grid areas as the original areal units and proceed similarly. Our goal is to find the regionalization with the minimum MVCAGE. As discussed by \citet{bradley2017regionalization}, finding the global minimum value of MVCAGE for all possible combinations of spatial units is computationally expensive. MVCAGE instead uses an alternate two-stage regionalization algorithm. The first stage of the regionalization algorithm uses the estimated values of the underlying spatial processes and applies any spatial clustering method to propose a set of candidate supports. For this, we specify lower and upper bounds on the number of possible regions or clusters. In the second stage, we compute MVCAGE for all these supports and retain the one with the minimum MVCAGE. This way, we find the particular regionalization with the minimum value of MVCAGE. We illustrated the application of the MVCAGE algorithm using a simulated data example and through an application to demographic and ocean color datasets.

There are several opportunities for future research based on the ideas proposed here. For example, extension of the MKLE to the spatio-temporal setting would provide a useful tool for investigating aggregation in both space and time. Another useful extension would be to consider different penalization on the number of regions. Lastly, an interesting direction of research would be to study the behavior of the global vs. local eigenfunctions of the corresponding MKLEs. 

\section*{Acknowledgments}
This article is released to inform interested parties of ongoing research and to encourage discussion. The views expressed on statistical issues are those of the authors and not those of the NSF or U.S. Census Bureau. This research was partially supported by the U.S. National Science Foundation (NSF) under NSF grants SES-1853096, NCSE-2215168, DMS-2310756. 
 
\begin{appendices}

%
{
\section{Proofs} \label{app2}

\noindent
\textbf{Proof of Proposition~\ref{acov}:} It follows from the Multivariate KLE theorem (e.g., see Proposition 4 of \citet{happ2018multivariate}, and the review in \citeauthor{daw2022overview}~\citeyear{daw2022overview}) that
\begin{equation}
	\Y(\cdot) = \sum_{k=1}^{\infty} \alpha_k \ppsi_{k}(\cdot){,\hspace{3pt}\mathrm{in}\hspace{3pt}\mathcal{L}_{2}}
\end{equation}
where P5 $\--$ P10 hold and each element $Y_{j}(\cdot)$ satisfies P1 $\--$ P5 via Proposition 5 of \citet{happ2018multivariate}. The claim that $	\Y\ar(\A) = \sum_{k=1}^{\infty} \alpha_k \ppsi\ar_{k}(\A)$ in $\mathcal{L}_{2}(\Omega)$ can be written as,
\begin{equation}
	\mathrm{SSE}(\A) \equiv\sum_{j=1}^{N}\mathrm{SSE}_{j}(\A)\equiv \sum_{j = 1}^{N} E\left\lbrace \left(Y_{j}\ar(\A) - \sum_{k = 1}^{M_{j}}\psi_{jk}\ar(\A)\alpha_{jk}\right)^{2}\right\rbrace
\end{equation}
\noindent
goes to zero as $M_{1},\ldots,M_{N}$ goes to infinity. Decompose $\mathrm{SSE}_{j}(\A)$ as,
\begin{equation}\label{eq:decompose}
	\mathrm{SSE}_{j}(\A) = E\left\lbrace Y_{j}\ar(\A)^{2}\right\rbrace + E\left\lbrace \left(\sum_{k = 1}^{M_{j}}\psi_{jk}\ar(\A)\alpha_{jk}\right)^{2}\right\rbrace - 2 E\left\lbrace Y_{j}\ar(\A)\sum_{k = 1}^{M_{j}}\psi_{jk}\ar(\A)\alpha_{jk}\right\rbrace,
\end{equation} 
\noindent 
where from some algebra and P10,
\begin{align}
	\nonumber
	E\left\lbrace Y_{j}\ar(\A)^{2}\right\rbrace &= \frac{1}{|A|^{2}}\int_{\A}\int_{\A} C_{jj}(\textbf{s},\textbf{r})d\textbf{s}d\textbf{r}\\
	\label{eq:firsttwoterms}
	E\left\lbrace \left(\sum_{k = 1}^{M_{j}}\psi_{jk}\ar(\A)\alpha_{jk}\right)^{2}\right\rbrace &= \frac{1}{|A|^{2}}\int_{\A}\int_{\A}\sum_{k = 1}^{M_{j}}\psi_{jk}(\textbf{s})\psi_{jk}(\textbf{r})\lambda_{jk}d\textbf{s}d\textbf{r}.
\end{align}
From Proposition 5 from \citet{happ2018multivariate} we have that the univariate processes of the multivariate KLE are univariate KLEs. This implies that the third term in (\ref{eq:decompose}) can be simplified by interpreting $\alpha_{jk}$ as a projection of $Y_{j}$ onto the eigenfunctions as follows,
\begin{align*}
E\left\lbrace Y_{j}\ar(\A)\sum_{k = 1}^{M_{j}}\psi_{jk}\ar(\A)\alpha_{jk}\right\rbrace &= 	E\left\lbrace \frac{1}{|A|^{2}}\sum_{k = 1}^{M_{j}}\int_{\A}\int_{\A}\int\psi_{jk}(\textbf{s})Y_{j}(\textbf{r})\alpha_{jk}d\textbf{s}d\textbf{r}\right\rbrace \\
&= 	E\left\lbrace \frac{1}{|A|^{2}}\sum_{k = 1}^{M_{j}}\int_{\A}\int_{\A}\int\psi_{jk}(\textbf{s})Y_{j}(\textbf{r})\int_{\Omega}Y_{j}(\textbf{u})\psi_{jk}(\textbf{u})d\textbf{u}d\textbf{s}d\textbf{r}\right\rbrace \\
&= 	E\left\lbrace \frac{1}{|A|^{2}}\sum_{k = 1}^{M_{j}}\int_{\A}\int_{\A}\int\psi_{jk}(\textbf{s})\int_{\Omega}Y_{j}(\textbf{r})Y_{j}(\textbf{u})\psi_{jk}(\textbf{u})d\textbf{u}d\textbf{s}d\textbf{r}\right\rbrace\\
&=	\frac{1}{|A|^{2}}\sum_{k = 1}^{M_{j}}\int_{\A}\int_{\A}\int\psi_{jk}(\textbf{s})\int_{\Omega}E\left\lbrace Y_{j}(\textbf{r})Y_{j}(\textbf{u})\right\rbrace\psi_{jk}(\textbf{u})d\textbf{u}d\textbf{s}d\textbf{r}\\
&=	\frac{1}{|A|^{2}}\sum_{k = 1}^{M_{j}}\int_{\A}\int_{\A}\int\psi_{jk}(\textbf{s})\int_{\Omega}{C}_{jj}(\textbf{r},\textbf{u})\psi_{jk}(\textbf{u})d\textbf{u}d\textbf{s}d\textbf{r}.
\end{align*}
From the Fredholm integral equation, this implies
\begin{align*}
	E\left\lbrace Y_{j}\ar(\A)\sum_{k = 1}^{M_{j}}\psi_{jk}\ar(\A)\alpha_{jk}\right\rbrace
= \frac{1}{|A|^{2}}\int_{\A}\int_{\A}\sum_{k = 1}^{M_{j}}\psi_{jk}(\textbf{s})\psi_{jk}(\textbf{r})\lambda_{jk}d\textbf{s}d\textbf{r}.
\end{align*}
Consequently,
\begin{equation*}
	\mathrm{SSE}_{j}(\A) = \frac{1}{|A|^{2}}\left\lbrace \int_{\A}\int_{\A}C_{jj}(\textbf{s},\textbf{r}) -\sum_{k = 1}^{M_{j}}\psi_{jk}(\textbf{s})\psi_{jk}(\textbf{r})\lambda_{jk} d\textbf{s}\textbf{r} \right\rbrace,
\end{equation*}
\noindent
so that by the univariate Mercer's theorem, which admits uniform covergence, we have that $\mathrm{SSE}_{j}(\A)$ coverges to zero. For fixed $N$ this implies that $\mathrm{SSE}(\A)$ converges to zero as $M_{1},\ldots, M_{N}$ each goes to infinity. This proves the multiscale multivariate KLE stated in Proposition 3.1.a. The proof of Statement 3.1.b follows immediately, since the right-hand-side of,
\begin{equation}
		\cov \big[\Y\ar(\A_i), \Y\ar(\A_j) \big] - \sum_{k=1}^{M} \lambda_k \,\ppsi\ar_{k}(\A_i) \, \ppsi\ar_{k}(\A_j)^{\top}= \frac{1}{|A|^{2}}\left\lbrace \int_{\A} \int_{\A} \mathbb{C}(\textbf{s},\textbf{r}) - \sum_{k=1}^{M} \lambda_k \,\ppsi_{k}(\textbf{s}) \, \ppsi_{k}(\textbf{r})^{\top}d\textbf{s}d\textbf{r}\right\rbrace,
\end{equation}
\noindent
converges to zero by the multivariate Mercer's theorem  \citep[Proposition 4 of][]{happ2018multivariate}.\\

\noindent
\textbf{Proof of Proposition \ref{prop1}:} Suppose $\psi_{\ell k}(\bm{x}_{j}) =\psi_{\ell k}\ar(\A_{j})$ for $\ell = 1,\ldots, N$, $j = 1,\ldots, n_{A}$, and $k = 1,\ldots$. For a given $\ell$ and $j$, it follows from Proposition 2 in \citet{bradley2017regionalization} that $f_{\ell}(Y_{\ell}(\bm{x}_{j})) \overset{a.s.}{=}f_{\ell}(Y_{\ell}\ar(\A_{j}))$ for any real-valued function $f_{\ell}$. Let $\bm{f}(\cdot) = (f_{1}(\cdot),\ldots, f_{N}(\cdot))$, so that $\bm{f}(\textbf{Y}^{x}) \overset{a.s.}{=} \bm{f}(\textbf{Y}^{\A})$. Similarly, $\bm{f}(\textbf{Y}^{x}) \overset{a.s.}{=} \bm{f}(\textbf{Y}^{\A})$ implies $f_{\ell}(Y_{\ell}(\bm{x}_{j})) \overset{a.s.}{=}f_{\ell}(Y_{\ell}\ar(\A_{j}))$ and from Proposition 2 of \citet{bradley2017regionalization} $\psi_{\ell k}(\bm{x}_{j}) =\psi_{\ell k}\ar(\A_{j})$ for $\ell = 1,\ldots, N$, $j = 1,\ldots, n_{A}$, and $k = 1,\ldots$, so that $\ppsi_{k}(\textbf{x}_{j}) = \ppsi_{j}\ar(\A_{j})$. This proves Proposition 3.2.a. The proof of Proposition 3.2.b is the same, but replaces $\textbf{x}_{j}$ with $\B_{j}$.\\

\noindent
\textbf{Proof of Proposition \ref{prop33}:} It follows from Proposition 5 of \citet{bradley2017regionalization} that
\begin{equation*}
	\big(\psi_{j1}^{OC}(\s), \ldots, \psi^{OC}_{j\widetilde{M}_j} (\s)\big) = \big(\theta_{j1}(\s), \ldots, \theta_{j\widetilde{M}_j} (\s)\big)\textbf{F}_{j},
\end{equation*}	
satisfies orthogonality, the Fredholm integral equations, and the associated truncated KLE has positive eigenvalues. Thus, it follows from Proposition 5 of \citet{happ2018multivariate} that the implied multivariate process is a multivariate KLE.
}

\section{Alternate Expressions of MVCAGE}
Recall that, 
\begin{align*}
    \cov \big(\Y(\s), \Y(\s)\big) &= \sum_{k=1}^{\infty} \lambda_k \,\ppsi_k(\s) \, \ppsi_k(\s)^{\top}, \\
    \cov \big(\Y\ar(\A), \Y\ar(\A)\big) &= \sum_{k=1}^{\infty} \lambda_k \,\ppsi_k\ar(\A) \, \ppsi_k\ar(\A)^{\top}.
\end{align*} 
Similarly, {it follows from the multivariate Mercer's theorem that we have} the following
\begin{align*}
     \cov\big(\Y(\s), \Y\ar(\A) \big) &= \frac{1}{|\A|} \int_{\bm{r} \in \A}   \cov \big(\Y(\s), \Y(\bm{r}) \big) \, d\bm{r} \\
     &= \frac{1}{|\A|} \int_{\bm{r} \in \A} \sum_{k=1}^{\infty} \lambda_k \ppsi_k(\s) \ppsi_k^{\top}(\bm{r}) \, d\bm{r} \\
     &= \sum_{k=1}^{\infty} \lambda_k \ppsi_k(\s) \frac{1}{|\A|} \int_{\bm{r} \in \A}\ppsi_k^{\top}(\bm{r}) \, d\bm{r} \\
     &= \sum_{k=1}^{\infty} \lambda_k \ppsi_k(\s) \ppsi_k\ar(\bm{r})^{\top}. \\
\end{align*}
Now, 
\begin{eqnarray*}
\sum_{k=1}^{\infty} \lambda_k \ppsi_k^{\top}(\s) \ppsi_k(\s)
&=& \trace\left(\sum_{k=1}^{\infty} \lambda_k{\ppsi_k^{\top}(\s) \ppsi_k(\s)} \right)\\
&=& \trace\left( \sum_{k=1}^{\infty} \lambda_k {\ppsi_k(\s) \ppsi_k(\s)^{\top}} \right)\\
& = &\trace\big( \cov(\Y(s)).     
\end{eqnarray*}
Similarly,   $\sum_{k=1}^{\infty} \lambda_k \ppsi_k\ar(\A)^{\top} \ppsi_k\ar(\A)$ = $\trace\big( \cov(\Y\ar(\A) \big)$. Then, under the squared-error loss, the MVCAGE loss becomes 
{
\begin{align*}
    MVCAGE(\A) &=  \mathbb{E}\left[\frac{1}{|\A|} \int_{\s \in \A} \sum_{k=1}^{\infty} \lambda_k \Big[\ppsi_k\ar(\A) - \ppsi_k(\s) \Big]^{\top} \Big[\ppsi_k\ar(\A) - \ppsi_k(\s) \Big]  \, d\s\vert \Z_{1},\ldots, \Z_{N}\right]\\
    &= \mathbb{E}\left[\sum_{k=1}^{\infty} \lambda_k \frac{1}{|\A|} \int_{\s \in \A} \Big[ \ppsi_k^{\top}(\s) \ppsi_k(\s) - 2 \ppsi_k^{\top}(\s) \ppsi_k\ar(\A) + \ppsi_k\ar(\A)^{\top} \ppsi_k\ar(\A) \Big]  \, d\s\vert \Z_{1},\ldots, \Z_{N}\right]\\
    &= \mathbb{E}\left[\frac{1}{|\A|} \int_{\s \in \A}  \trace\Big[ \cov \big(\Y(\s), \Y(\s) \big) - 2\cov \big(\Y(\s), \Y\ar(\A) \big) + \cov \big(\Y\ar(\A), \Y\ar(\A) \big) \Big]  \, d\s\vert \Z_{1},\ldots, \Z_{N}\right]\\
    &= \mathbb{E}\left[\frac{1}{|\A|} \int_{\s \in \A}  \trace\Big[ \cov\big( \Y(\s) - \Y\ar(\A), \Y(\s) - \Y\ar(\A)\big)  \Big]  \, d\s\vert \Z_{1},\ldots, \Z_{N}\right] \\
    &= \mathbb{E}\left[\frac{1}{|\A|} \int_{\s \in \A}  \trace \, \mathbb{E}\Big[\Big( \Y(\s) - \Y\ar(\A) \Big) \Big(\Y(\s) - \Y\ar(\A)\Big)^{\top} \Big] \, d\s\vert \Z_{1},\ldots, \Z_{N}\right].
\end{align*}
}
The other part follows straightforwardly by noting that $\cov \big(\Y(\s), \Y\ar(\A) \big) = \frac{1}{|\A|} \int_{\bm{r} \in \A}\cov \big(\Y(\s), \Y(\bm{r} \big) \, d\bm{r}$.

{
	This leads to the following ANOVA type decomposition of the multivariate CAGE:
	\begin{equation}
		MVCAGE(\A) = \mathbb{E}\left[\frac{1}{|\A|} \int_{\s \in \A}\trace \Big[ \cov \{\Y(\s) \}\Big]d\s -  \trace \Big[ \cov  \{\Y\ar(\A) \}\Big]\vert \Z_{1},\ldots, \Z_{N}\right].  
	\end{equation}
Thus, the first term represents the average variability of the multivariate spatial process on the point-referenced scale, whereas the second term represents the variability of the multivariate spatial process on the areal/regional scale.
}

\section{The General MVCAGE}
Although we used the weighted square error loss, we can generalize the above notion by using any convex loss function for multivariate data, e.g., Manhattan loss, KL divergence, sup-norm loss etc. Under a loss function $\mathcal{L}$, the generalized version of MVCAGE (GMVCAGE) can be expressed as
\begin{align*}
	{\mathrm{GMVCAGE}(\A) = \frac{1}{|\A|} \int_{\s \in \A} E\left[\sum_{k=1}^{\infty} \mathcal{L} \big[\ppsi_k(\A) , \ppsi_k(\s); \lambda_k \big]\vert \Z_{1},\ldots, \Z_{N}\right] \, d\s,}
\end{align*}
One can choose any problem-specific loss function here; e.g., \citet{daw2022supervised} used $\mathcal{L}_1$ divergence within the framework of a minimum spanning tree and then applied the $\mathcal{L}_2$ loss to get the final partition for a univariate CAGE example. Specifically, the $\mathcal{L}_1$ loss is often advantageous for outliers and robustness. Univariate CAGE in \citet{bradley2017regionalization} used a weighted squared error loss, with the $k$-th eigenfunction weighted by the eigenvalue $\lambda_k$.

\section{Block Gibbs Sampling for the MVOC Model} \label{computationgibbs}
%
%
Denote $\widehat{\bm{\nu}}_j$ as the ordinary least-square estimate with $Z_j(\s)$ as response and $\phi_{jk}(\s)$ as predictor. Denote its $k$-th element as $\widehat{\nu}_{jk}$. Also denote the $g$-sum of square as $\operatorname{SSR}_j = \sum_{i=1}^n \big[Z_j(\s_i) - \mu_j - \sum_{k=1}^{M_j}\frac{g}{g+1}\phi_{jk}(\s_i) \widehat{\nu}_{jk} \big]^2$. Then, one can draw posterior samples of the hyperparameters using a block Gibbs sampler, where the hyperparameters are sequentially sampled from each of their specific full conditional distributions as follows
\begin{eqnarray*}
    \pi\big(\mu_j | \cdot \big) &\sim& \mathcal{N} \bigg(\frac{1}{n} \sum_{k=1}^n \big[Z_j(\s_i) - \sum_{k=1}^{M_j}\frac{g}{g+1}\phi_{jk}(\s_i) \widehat{\nu}_{jk} \big], \frac{\sigma_j^2}{n}\bigg), \\
    \pi\big(\sigma_j^2 | \cdot \big) &\sim& \operatorname{Inverse} \operatorname{Gamma} \bigg(\frac{n}{2}, \, \frac{\operatorname{SSR}_j}{2} \bigg), \\
    \pi\big(\bm{\nu}_j | \cdot \big) &\sim& \mathcal{N} \Big( \frac{g}{g+1} \widehat{\bm{\nu}}_j, \frac{g}{g+1} \sigma_j^2 \big(\bm{\Phi}_{j}^{\top} \bm{\Phi}_{j} \big)^{-1} \Big).
\end{eqnarray*}
\noindent One first samples from the full conditional distribution $\pi\big(\sigma_j^2 | \YY \big)$. Then, conditioned on the sampled $\sigma^2_j$, we can draw samples of $\mu_j$ and $\bm{\nu}_j$ from their conditional distributions. Then, using Monte Carlo integration, we can integrate out $\mu_j$ and $\sigma^2_j$ to get $\pi(\bm{\nu}_j | \YY )$.

\section{Alternative Modeling and MKLE Computation} \label{klecmp2}
We also demonstrate a few examples using alternative covariance modeling techniques. In particular, for cross-covariance models, we employ the linear model for co-regionalization (LMC) \citep{journel1978mining, goulard1992linear} to estimate the full covariance matrix of $\bm{\Z}\supvec$.  Interested readers can see \citet{wackernagel1989overview} for a survey on co-regionalization.  Specifically, we use the co-kriging procedure from the \texttt{R} package \texttt{gstat} \citep{pebesma2015package, rossiter2007co}. For a bivariate process, the LMC first models one of the spatial processes (say $Z_1(\s)$) using the traditional geospatial model \citep{cressie2015statistics}. Then, the second spatial process is built as a conditional model using $Z_1(\s)$, i.e., we estimate the distribution of $\pi \big[Z_2(\s) | Z_1(\s) \big]$. We use this approach to estimate the underlying parameters of the covariance matrix (cf. Section~\ref{oceanex} for an example).

We consider another example of empirical MKLE with the assumption of repeated data. When sufficient multiple replications of the multivariate variables are available, i.e., more samples than the number of parameters are available, one can compute the empirical covariance matrix of the data as 
\begin{align} \label{empmkle}
	\widehat{\CC}_{ij}(\s, \bm{r}) = \frac{1}{\widetilde{n}} \sum_{k=1}^{\widetilde{n}} \big[Z_{i, k}(\s) - \bar{Z}_{i}(\s)\big] \big[Z_{j, k}(\bm{r}) - \bar{Z}_{j}(\bm{r})\big]^{\top},    
\end{align}
\noindent where  $Z_{j, k}$ is the $k$-th replication of the $j$-th process and $\bar{Z}_{j}(\s) = \frac{1}{\widetilde{n}} \sum_{k=1}^{\widetilde{n}} Z_{j, k}(\s)$ (cf. Section~\ref{mvfda} for an example). 

If we first estimate the form covariance functions $\widehat{\CC}_{jk}$, we need to compute the MKLE from them. We use the Galerkin approach \citep{ghanem1991stochastic} here, which is a similar routine as used in the MVOC basis construction. To compute the KLE for the $j$-th univariate process, we start with a set of orthonormal basis functions  $\{\phi_{jk}(\s)\}_{k=1}^{\widetilde{M}_j}$. Such families of basis functions include Fourier, Haar, spline, polynomial basis functions, and many others. Given the estimate $\widehat{\CC}_{jj}(\s, \bm{r})$, one can compute the matrix $\mathbf{U}_j$ with elements $U^j_{k \ell}$ as 
\begin{align*}
	U^j_{k \ell} = \int_{s} \int_{\bm{r}} \widehat{\CC}_{jj}(\s, \bm{r}) \phi_{jk}(\s) \phi_{j\ell}(\bm{r}) \, d\s \, d\bm{r}.
\end{align*}
Assuming the positive definiteness of $\mathbf{U}_j$, we can perform an eigendecomposition on $\mathbf{U}_j$ to get eigenvalues $\lambda_{j1} \geqslant \cdots \geqslant \lambda_{jM_j} $ and eigenvectors $\bm{f}_{j1}, \ldots, \bm{f}_{jM_j}$. Then, the $k$-th eigenfunction of $\widehat{\CC}_{jj}$ is given by
\begin{align*}
	\psi_{jk}(\s) = \big(\bm{\phi}_{j1}(\s), \ldots,  \bm{\phi}_{j\widetilde{M}_j}(\s)\big)^{\top} \bm{f}_{jk}.
\end{align*}
Proof of the correctness of this derivation is similar to the ones in Appendix~\ref{app2}.

\begin{remark}
	Note that if the number of replications are less than the rank of the covariance matrix, one can use a similar technique as in the MVOC approach to approximate the first few eigenpairs. In that case, using the OC basis, we assume the following basis expansion model
	\begin{align*}
		Z_{j, k}(\s) = \mu_j + \sum_{i=1}^{\widetilde{M}_j} \phi_{ji}(\s) \nu^k_{ji} + \epsilon_j(\s),
	\end{align*}
	where, like before, $Z_{j, k}(\s)$ is the $k$-th replication of the $j$-th process, and $\nu^k_{ji}$ is the $i$-th basis expansion for the $k$-th replication of the $j$-th process. Then, we can compute the covariance matrix $\bm{K}_{jm}$ between the expansion coefficients with its $k,\ell$-th element given by $\frac{1}{\widetilde{n}} \sum_{i=1}^{\widetilde{n}} \big( \nu^i_{jk} - \bar{\nu}_{jk} \big) \big( \nu^i_{m\ell} - \bar{\nu}_{m\ell}\big)^{\top}$, where $\bar{\nu}_{jk} = \frac{1}{\widetilde{n}} \sum_{i=1}^{\widetilde{n}} \nu^i_{jk}$. Following a similar procedure as in the MVOC basis, we decorrelate the block matrix $\K$ with blocks $\bm{K}_{jk}$ to find the multivariate eigenfunctions and eigenvalues.
\end{remark}

\section{Clustering Considerations}\label{app:cluster}
The primary tool behind spatial regionalization is spatial clustering. Unlike nonspatial clustering, spatial clustering needs to account for the spatial proximity of the units to identify spatially contiguous clusters. Even then, many spatial clustering algorithms cannot guarantee explicit spatial contiguity and may require additional ad-hoc techniques \citep[see][for a review]{duque2007supervised}. In the context of MVCAGE, we explored some of these choices, namely spatial $k$-means \citep{alexandrov2011efficient}, hierarchical geographical clustering \citep{guha1998cure, chavent2018clustgeo}, and minimum spanning trees \citep{assunccao2006efficient}. 

We first discuss spatial $k$-means and hierarchical geographical clustering (HGC) approaches, which are straightforward extensions of their respective nonspatial  versions. Suppose we are interested in clustering a random vector $\Y(\cdot)$ using spatial $k$-means and HGC. We first create a new variable as $\bm{\mathcal{Y}}(\s) = \big(\gamma \Y(\s), (1 - \gamma)\s \big)$, where the hyperparameter $\gamma$ balances the ratio of effects of $\Y(\s)$ and the spatial proximity. We apply the respective clustering algorithms (i.e., nonspatial $k$-means and hierarchical clustering) on $\bm{\mathcal{Y}}$ to get the spatial clusters. Finally, units within the same cluster are joined to create a partition or regionalization of $\SSS$. In practice, explicitly local partitions are often obtained in an ad-hoc fashion. For example, one can carefully tune the balancing parameter $\gamma$ such that only local clusters are formed. Alternatively, one can create naive partitions using the nonspatial clustering methods and then subsequently partition the clusters with non-neighboring locations into multiple smaller clusters. Although these methods come with additional concerns about these ad-hoc techniques, they are quite commonly used due to their simple interpretations and available implementations \citep[such as ][]{chavent2018clustgeo}.

Alternatively, the minimum spanning tree (MST)  \citep{assunccao2006efficient, teixeira2019bayesian, luo2021bayesian} is a method that yields an explicitly spatial regionalization, i.e., it only allows neighboring units in a single cluster. We first use the new response variables $\mathcal{Y}(\s)$ to form a connected graph over the spatial domain $\SSS$. This represents 
the spatial neighborhood as a connectivity graph, where only spatial neighbors are connected with each other with unique ``edges''. An edge weight is used to reflect the variable under study that is being partitioned. Subsequently, edges with higher weights are removed to yield clusters, within which units share at least one neighboring unit by construction. See \citet{daw2022supervised} for a univariate spatial regionalization methodology using MST, where the CAGE loss was incorporated within the framework of the MST. 

Here we use a hybrid idea of the spatial clustering techniques, borrowing the HGC method from \citet{bradley2017regionalization} and the clustering technique in \citet{daw2022overview}. The MST methodology in \citet{daw2022supervised} is applied on the space of the KLE, which by construction generate explicitly local clusters. However, we found in our studies that the naive clustering algorithms, such as $k$-means and HGC with Ward linkage, work reasonably well to produce spatially contiguous partitions. Additionally, the HGC with Ward linkage mitigates the computational and boundary issues of the MST, which had to be mitigated in \citet{daw2022overview} by imposing additional constraints. We also note that MST is a specific case of HGC with single linkage method. Since Ward linkage has been argued to be advantageous and a reasonable choice for agglomerative clustering, here we chose the HGC with Ward clustering to generate the spatial partitions.

\end{appendices}
\bibliographystyle{elsarticle-num-names} 
\bibliography{mvcage}

\begin{thebibliography}{49}
\expandafter\ifx\csname natexlab\endcsname\relax\def\natexlab#1{#1}\fi
\providecommand{\url}[1]{\texttt{#1}}
\providecommand{\href}[2]{#2}
\providecommand{\path}[1]{#1}
\providecommand{\DOIprefix}{doi:}
\providecommand{\ArXivprefix}{arXiv:}
\providecommand{\URLprefix}{URL: }
\providecommand{\Pubmedprefix}{pmid:}
\providecommand{\doi}[1]{\href{http://dx.doi.org/#1}{\path{#1}}}
\providecommand{\Pubmed}[1]{\href{pmid:#1}{\path{#1}}}
\providecommand{\bibinfo}[2]{#2}
\ifx\xfnm\relax \def\xfnm[#1]{\unskip,\space#1}\fi
\bibitem[{Openshaw and Taylor(1979)}]{openshaw1979million}
\bibinfo{author}{S.~Openshaw}, \bibinfo{author}{P.~Taylor},
\newblock \bibinfo{title}{A million or so correlation coefficients: Three experiments on the modifiable areal unit problem},
\newblock \bibinfo{journal}{Statistical applications in the spatial sciences} \bibinfo{volume}{29} (\bibinfo{year}{1979}) \bibinfo{pages}{127--144}.
\bibitem[{Gotway and Young(2002)}]{gotway2002combining}
\bibinfo{author}{C.~A. Gotway}, \bibinfo{author}{L.~J. Young},
\newblock \bibinfo{title}{Combining incompatible spatial data},
\newblock \bibinfo{journal}{Journal of the American Statistical Association} \bibinfo{volume}{97} (\bibinfo{year}{2002}) \bibinfo{pages}{632--648}.
\bibitem[{Banerjee et~al.(2004)Banerjee, Carlin, and Gelfand}]{banerjee2004hierarchical}
\bibinfo{author}{S.~Banerjee}, \bibinfo{author}{B.~P. Carlin}, \bibinfo{author}{A.~E. Gelfand}, \bibinfo{title}{Hierarchical modeling and analysis for spatial data}, \bibinfo{publisher}{CRC press}, \bibinfo{year}{2004}.
\bibitem[{Wikle and Berliner(2005)}]{wikle2005combining}
\bibinfo{author}{C.~K. Wikle}, \bibinfo{author}{L.~M. Berliner},
\newblock \bibinfo{title}{Combining information across spatial scales},
\newblock \bibinfo{journal}{Technometrics} \bibinfo{volume}{47} (\bibinfo{year}{2005}) \bibinfo{pages}{80--91}.
\bibitem[{Robinson(2009)}]{robinson2009ecological}
\bibinfo{author}{W.~S. Robinson},
\newblock \bibinfo{title}{Ecological correlations and the behavior of individuals},
\newblock \bibinfo{journal}{International Journal of Epidemiology} \bibinfo{volume}{38} (\bibinfo{year}{2009}) \bibinfo{pages}{337--341}.
\bibitem[{Bradley et~al.(2017)Bradley, Wikle, and Holan}]{bradley2017regionalization}
\bibinfo{author}{J.~R. Bradley}, \bibinfo{author}{C.~K. Wikle}, \bibinfo{author}{S.~H. Holan},
\newblock \bibinfo{title}{Regionalization of multiscale spatial processes by using a criterion for spatial aggregation error},
\newblock \bibinfo{journal}{Journal of the Royal Statistical Society: Series B (Statistical Methodology)} \bibinfo{volume}{79} (\bibinfo{year}{2017}) \bibinfo{pages}{815--832}.
\bibitem[{Zhou and Bradley(2023)}]{zhou2023bayesian}
\bibinfo{author}{S.~Zhou}, \bibinfo{author}{J.~R. Bradley},
\newblock \bibinfo{title}{Bayesian hierarchical modeling for bivariate multiscale spatial data with application to blood test monitoring},
\newblock \bibinfo{journal}{arXiv preprint arXiv:2310.13580}  (\bibinfo{year}{2023}).
\bibitem[{Cressie(1993)}]{cressie1993statistics}
\bibinfo{author}{N.~Cressie}, \bibinfo{title}{Statistics for spatial data}, \bibinfo{publisher}{John Wiley \& Sons}, \bibinfo{year}{1993}.
\bibitem[{Cressie and Wikle(2011)}]{cressie2015statistics}
\bibinfo{author}{N.~Cressie}, \bibinfo{author}{C.~K. Wikle}, \bibinfo{title}{Statistics for Spatio-temporal Data}, \bibinfo{publisher}{John Wiley \& Sons}, \bibinfo{year}{2011}.
\bibitem[{Waller(2004)}]{waller2004applied}
\bibinfo{author}{L.~Waller},
\newblock \bibinfo{title}{Applied spatial statistics for public health data},
\newblock \bibinfo{journal}{Willey \& Sons}  (\bibinfo{year}{2004}).
\bibitem[{Qu et~al.(2021)Qu, Bradley, and Niu}]{qu2021boundary}
\bibinfo{author}{K.~Qu}, \bibinfo{author}{J.~R. Bradley}, \bibinfo{author}{X.~Niu},
\newblock \bibinfo{title}{Boundary detection using a bayesian hierarchical model for multiscale spatial data},
\newblock \bibinfo{journal}{Technometrics} \bibinfo{volume}{63} (\bibinfo{year}{2021}) \bibinfo{pages}{64--76}.
\bibitem[{Daw et~al.(2022)Daw, Simpson, Wikle, Holan, and Bradley}]{daw2022overview}
\bibinfo{author}{R.~Daw}, \bibinfo{author}{M.~Simpson}, \bibinfo{author}{C.~K. Wikle}, \bibinfo{author}{S.~H. Holan}, \bibinfo{author}{J.~R. Bradley},
\newblock \bibinfo{title}{An overview of univariate and multivariate {K}arhunen-{L}o{\`e}ve expansions in statistics},
\newblock \bibinfo{journal}{Journal of the Indian Society for Probability and Statistics} \bibinfo{volume}{23} (\bibinfo{year}{2022}) \bibinfo{pages}{285--326}.
\bibitem[{Openshaw(1984)}]{openshaw1984modifiable}
\bibinfo{author}{S.~Openshaw},
\newblock \bibinfo{title}{The modifiable areal unit problem},
\newblock \bibinfo{journal}{Concepts and Techniques in Modern Geography} \bibinfo{volume}{38} (\bibinfo{year}{1984}) \bibinfo{pages}{41--57}.
\bibitem[{Anselin(1988)}]{anselin1988spatial}
\bibinfo{author}{L.~Anselin}, \bibinfo{title}{Spatial econometrics: methods and models}, \bibinfo{publisher}{Springer-Verlag}, \bibinfo{year}{1988}.
\bibitem[{Bailey and Gatrell(1995)}]{bailey1995interactive}
\bibinfo{author}{T.~C. Bailey}, \bibinfo{author}{A.~C. Gatrell}, \bibinfo{title}{Interactive spatial data analysis}, \bibinfo{publisher}{Longman Scientific \& Technical}, \bibinfo{year}{1995}.
\bibitem[{Obled and Creutin(1986)}]{obled1986some}
\bibinfo{author}{C.~Obled}, \bibinfo{author}{J.~Creutin},
\newblock \bibinfo{title}{Some developments in the use of empirical orthogonal functions for mapping meteorological fields},
\newblock \bibinfo{journal}{Journal of Applied Meteorology and Climatology} \bibinfo{volume}{25} (\bibinfo{year}{1986}) \bibinfo{pages}{1189--1204}.
\bibitem[{Cho et~al.(2013)Cho, Venturi, and Karniadakis}]{cho2013karhunen}
\bibinfo{author}{H.~Cho}, \bibinfo{author}{D.~Venturi}, \bibinfo{author}{G.~Karniadakis},
\newblock \bibinfo{title}{{K}arhunen-{L}o{\`e}ve expansion for multi-correlated stochastic processes},
\newblock \bibinfo{journal}{Probabilistic Engineering Mechanics} \bibinfo{volume}{34} (\bibinfo{year}{2013}) \bibinfo{pages}{157--167}.
\bibitem[{Happ and Greven(2018)}]{happ2018multivariate}
\bibinfo{author}{C.~Happ}, \bibinfo{author}{S.~Greven},
\newblock \bibinfo{title}{Multivariate functional principal component analysis for data observed on different (dimensional) domains},
\newblock \bibinfo{journal}{Journal of the American Statistical Association} \bibinfo{volume}{113} (\bibinfo{year}{2018}) \bibinfo{pages}{649--659}.
\bibitem[{Karhunen(1947)}]{Karhunen1947}
\bibinfo{author}{K.~Karhunen},
\newblock \bibinfo{title}{Zur spektraltheorie stochastischer prozesse},
\newblock \bibinfo{journal}{Annales Academiae Scientiarum Fennicae} \bibinfo{volume}{37} (\bibinfo{year}{1947}) \bibinfo{pages}{1--79}.
\bibitem[{Lo{\`e}ve(1945)}]{Loeve1945}
\bibinfo{author}{M.~Lo{\`e}ve}, \bibinfo{title}{Probability Theory}, volume~\bibinfo{volume}{1}, \bibinfo{publisher}{Van Nostrand Princeton}, \bibinfo{year}{1945}.
\bibitem[{Aronszajn(1950)}]{aronszajn1950theory}
\bibinfo{author}{N.~Aronszajn},
\newblock \bibinfo{title}{Theory of reproducing kernels},
\newblock \bibinfo{journal}{Transactions of the American Mathematical Society} \bibinfo{volume}{68} (\bibinfo{year}{1950}) \bibinfo{pages}{337--404}.
\bibitem[{Resnick(2013)}]{resnick2013probability}
\bibinfo{author}{S.~I. Resnick}, \bibinfo{title}{A probability path}, \bibinfo{publisher}{Springer Science \& Business Media}, \bibinfo{year}{2013}.
\bibitem[{Sethuraman(1994)}]{sethuraman1994constructive}
\bibinfo{author}{J.~Sethuraman},
\newblock \bibinfo{title}{A constructive definition of dirichlet priors},
\newblock \bibinfo{journal}{Statistica sinica}  (\bibinfo{year}{1994}) \bibinfo{pages}{639--650}.
\bibitem[{Schmidt-Hieber(2021)}]{schmidt2021kolmogorov}
\bibinfo{author}{J.~Schmidt-Hieber},
\newblock \bibinfo{title}{The kolmogorov--arnold representation theorem revisited},
\newblock \bibinfo{journal}{Neural networks} \bibinfo{volume}{137} (\bibinfo{year}{2021}) \bibinfo{pages}{119--126}.
\bibitem[{Journel and Huijbregts(1978)}]{journel1978mining}
\bibinfo{author}{A.~G. Journel}, \bibinfo{author}{C.~J. Huijbregts},
\newblock \bibinfo{title}{Mining geostatistics},
\newblock \bibinfo{journal}{Journal of International Association for Mathematical Geology} \bibinfo{volume}{10} (\bibinfo{year}{1978}) \bibinfo{pages}{395--424}.
\bibitem[{Goulard and Voltz(1992)}]{goulard1992linear}
\bibinfo{author}{M.~Goulard}, \bibinfo{author}{M.~Voltz},
\newblock \bibinfo{title}{Linear coregionalization model: tools for estimation and choice of cross-variogram matrix},
\newblock \bibinfo{journal}{Mathematical Geology} \bibinfo{volume}{24} (\bibinfo{year}{1992}) \bibinfo{pages}{269--286}.
\bibitem[{Gneiting et~al.(2010)Gneiting, Kleiber, and Schlather}]{gneiting2010matern}
\bibinfo{author}{T.~Gneiting}, \bibinfo{author}{W.~Kleiber}, \bibinfo{author}{M.~Schlather},
\newblock \bibinfo{title}{Mat{\'e}rn cross-covariance functions for multivariate random fields},
\newblock \bibinfo{journal}{Journal of the American Statistical Association} \bibinfo{volume}{105} (\bibinfo{year}{2010}) \bibinfo{pages}{1167--1177}.
\bibitem[{Berger(2013)}]{berger2013statistical}
\bibinfo{author}{J.~O. Berger}, \bibinfo{title}{Statistical decision theory and Bayesian analysis}, \bibinfo{publisher}{Springer Science \& Business Media}, \bibinfo{year}{2013}.
\bibitem[{Daw and Wikle(2022)}]{daw2022supervised}
\bibinfo{author}{R.~Daw}, \bibinfo{author}{C.~K. Wikle},
\newblock \bibinfo{title}{Supervised spatial regionalization using the {K}arhunen-{L}o{\`e}ve expansion and minimum spanning trees},
\newblock \bibinfo{journal}{Journal of Data Science} \bibinfo{volume}{20} (\bibinfo{year}{2022}) \bibinfo{pages}{566--584}.
\bibitem[{Genton and Kleiber(2015)}]{genton2015cross}
\bibinfo{author}{M.~G. Genton}, \bibinfo{author}{W.~Kleiber},
\newblock \bibinfo{title}{Cross-covariance functions for multivariate geostatistics},
\newblock \bibinfo{journal}{Statistical Science} \bibinfo{volume}{30} (\bibinfo{year}{2015}) \bibinfo{pages}{147 -- 163}. \URLprefix \url{https://doi.org/10.1214/14-STS487}. \DOIprefix\doi{10.1214/14-STS487}.
\bibitem[{Zellner(1986)}]{zellner1986introduction}
\bibinfo{author}{A.~Zellner}, \bibinfo{title}{An introduction to {B}ayesian inference in {E}conometrics}, \bibinfo{publisher}{John Wiley \& Sons}, \bibinfo{year}{1986}.
\bibitem[{Agliari and Parisetti(1988)}]{li2021zellner}
\bibinfo{author}{A.~Agliari}, \bibinfo{author}{C.~C. Parisetti},
\newblock \bibinfo{title}{A-g reference informative prior: A note on {Z}ellner's g-prior},
\newblock \bibinfo{journal}{Journal of the Royal Statistical Society: Series D (The Statistician)} \bibinfo{volume}{37} (\bibinfo{year}{1988}) \bibinfo{pages}{271--275}.
\bibitem[{Wackernagel et~al.(1989)Wackernagel, Petitgas, and Touffait}]{wackernagel1989overview}
\bibinfo{author}{H.~Wackernagel}, \bibinfo{author}{P.~Petitgas}, \bibinfo{author}{Y.~Touffait},
\newblock \bibinfo{title}{Overview of methods for coregionalization analysis},
\newblock in: \bibinfo{booktitle}{Geostatistics: Proceedings of the Third International Geostatistics Congress September 5--9, 1988, Avignon, France}, \bibinfo{organization}{Springer}, \bibinfo{year}{1989}, pp. \bibinfo{pages}{409--420}.
\bibitem[{Pebesma(2004)}]{pebesma2015package}
\bibinfo{author}{E.~J. Pebesma},
\newblock \bibinfo{title}{Multivariable geostatistics in {S}: the gstat package},
\newblock \bibinfo{journal}{Computers \& Geosciences} \bibinfo{volume}{30} (\bibinfo{year}{2004}) \bibinfo{pages}{683--691}.
\bibitem[{Gräler et~al.(2016)Gräler, Pebesma, and Heuvelink}]{gstat2}
\bibinfo{author}{B.~Gräler}, \bibinfo{author}{E.~Pebesma}, \bibinfo{author}{G.~Heuvelink},
\newblock \bibinfo{title}{Spatio-temporal interpolation using gstat},
\newblock \bibinfo{journal}{The R Journal} \bibinfo{volume}{8} (\bibinfo{year}{2016}) \bibinfo{pages}{204--218}. \URLprefix \url{https://journal.r-project.org/archive/2016/RJ-2016-014/index.html}.
\bibitem[{Rossiter(2007)}]{rossiter2007co}
\bibinfo{author}{D.~Rossiter},
\newblock \bibinfo{title}{Co-kriging with the gstat package of the {R} environment for statistical computing},
\newblock \bibinfo{journal}{Web: http://www. itc. nl/rossiter/teach/R/R ck. pdf}  (\bibinfo{year}{2007}).
\bibitem[{Hastie et~al.(2009)Hastie, Tibshirani, Friedman, and Friedman}]{hastie2009elements}
\bibinfo{author}{T.~Hastie}, \bibinfo{author}{R.~Tibshirani}, \bibinfo{author}{J.~H. Friedman}, \bibinfo{author}{J.~H. Friedman}, \bibinfo{title}{The elements of statistical learning: data mining, inference, and prediction}, volume~\bibinfo{volume}{2}, \bibinfo{publisher}{Springer}, \bibinfo{year}{2009}.
\bibitem[{Guttorp and Gneiting(2006)}]{guttorp2006studies}
\bibinfo{author}{P.~Guttorp}, \bibinfo{author}{T.~Gneiting},
\newblock \bibinfo{title}{Studies in the history of probability and {s}tatistics {XLIX} on the mat{\'e}rn correlation family},
\newblock \bibinfo{journal}{Biometrika} \bibinfo{volume}{93} (\bibinfo{year}{2006}) \bibinfo{pages}{989--995}.
\bibitem[{Leeds et~al.(2014)Leeds, Wikle, and Fiechter}]{leeds2014emulator}
\bibinfo{author}{W.~B. Leeds}, \bibinfo{author}{C.~K. Wikle}, \bibinfo{author}{J.~Fiechter},
\newblock \bibinfo{title}{Emulator-assisted reduced-rank ecological data assimilation for nonlinear multivariate dynamical spatio-temporal processes},
\newblock \bibinfo{journal}{Statistical Methodology} \bibinfo{volume}{17} (\bibinfo{year}{2014}) \bibinfo{pages}{126--138}.
\bibitem[{Wikle et~al.(2013)Wikle, Milliff, Herbei, and Leeds}]{wikle2013modern}
\bibinfo{author}{C.~K. Wikle}, \bibinfo{author}{R.~F. Milliff}, \bibinfo{author}{R.~Herbei}, \bibinfo{author}{W.~B. Leeds},
\newblock \bibinfo{title}{Modern statistical methods in oceanography: a hierarchical perspective},
\newblock \bibinfo{journal}{Statistical Science} \bibinfo{volume}{28} (\bibinfo{year}{2013}) \bibinfo{pages}{466--486}.
\bibitem[{Haidvogel et~al.(2010)Haidvogel, Arango, Budgell, Cornuelle, Curchitser, Di~Lorenzo, Fennel, Geyer, Hermann, Lanerolle, Levin, McWilliams, Miller, Moore, Powell, Shchepetkin, Sherwood, Signell, Warner, and Wilkin}]{doi:10.1175/2010MWR3421.1}
\bibinfo{author}{D.~B. Haidvogel}, \bibinfo{author}{H.~G. Arango}, \bibinfo{author}{W.~P. Budgell}, \bibinfo{author}{B.~D. Cornuelle}, \bibinfo{author}{E.~Curchitser}, \bibinfo{author}{E.~Di~Lorenzo}, \bibinfo{author}{K.~Fennel}, \bibinfo{author}{W.~R. Geyer}, \bibinfo{author}{A.~J. Hermann}, \bibinfo{author}{L.~Lanerolle}, \bibinfo{author}{J.~Levin}, \bibinfo{author}{J.~C. McWilliams}, \bibinfo{author}{A.~J. Miller}, \bibinfo{author}{A.~M. Moore}, \bibinfo{author}{T.~M. Powell}, \bibinfo{author}{A.~F. Shchepetkin}, \bibinfo{author}{C.~R. Sherwood}, \bibinfo{author}{R.~P. Signell}, \bibinfo{author}{J.~C. Warner}, \bibinfo{author}{J.~Wilkin},
\newblock \bibinfo{title}{Ocean forecasting in terrain-following coordinates: formulation and skill assessment of the regional ocean modeling system},
\newblock \bibinfo{journal}{Monthly Weather Review} \bibinfo{volume}{138} (\bibinfo{year}{2010}) \bibinfo{pages}{619--646}. \URLprefix \url{https://doi.org/10.1175/2010MWR3421.1}. \DOIprefix\doi{10.1175/2010MWR3421.1}.
\bibitem[{Ghanem and Spanos(1991)}]{ghanem1991stochastic}
\bibinfo{author}{R.~Ghanem}, \bibinfo{author}{P.~D. Spanos},
\newblock \bibinfo{title}{Stochastic finite elements: a spectral approach},
\newblock \bibinfo{journal}{Springer Science \& Business Media}  (\bibinfo{year}{1991}).
\bibitem[{Duque et~al.(2007)Duque, Ramos, and Suri{\~n}ach}]{duque2007supervised}
\bibinfo{author}{J.~C. Duque}, \bibinfo{author}{R.~Ramos}, \bibinfo{author}{J.~Suri{\~n}ach},
\newblock \bibinfo{title}{Supervised regionalization methods: a survey},
\newblock \bibinfo{journal}{International Regional Science Review} \bibinfo{volume}{30} (\bibinfo{year}{2007}) \bibinfo{pages}{195--220}.
\bibitem[{Alexandrov and Kobarg(2011)}]{alexandrov2011efficient}
\bibinfo{author}{T.~Alexandrov}, \bibinfo{author}{J.~H. Kobarg},
\newblock \bibinfo{title}{Efficient spatial segmentation of large imaging mass spectrometry datasets with spatially aware clustering},
\newblock \bibinfo{journal}{Bioinformatics} \bibinfo{volume}{27} (\bibinfo{year}{2011}) \bibinfo{pages}{i230--i238}.
\bibitem[{Guha et~al.(1998)Guha, Rastogi, and Shim}]{guha1998cure}
\bibinfo{author}{S.~Guha}, \bibinfo{author}{R.~Rastogi}, \bibinfo{author}{K.~Shim},
\newblock \bibinfo{title}{{CURE}: An efficient clustering algorithm for large databases},
\newblock in: \bibinfo{booktitle}{ACM SIGMOD Conference}, \bibinfo{organization}{ACM}, \bibinfo{year}{1998}.
\bibitem[{Chavent et~al.(2018)Chavent, Kuentz-Simonet, Labenne, and Saracco}]{chavent2018clustgeo}
\bibinfo{author}{M.~Chavent}, \bibinfo{author}{V.~Kuentz-Simonet}, \bibinfo{author}{A.~Labenne}, \bibinfo{author}{J.~Saracco},
\newblock \bibinfo{title}{Clustgeo: an {R} package for hierarchical clustering with spatial constraints},
\newblock \bibinfo{journal}{Computational Statistics} \bibinfo{volume}{33} (\bibinfo{year}{2018}) \bibinfo{pages}{1799--1822}.
\bibitem[{Assun{\c{c}}{\~a}o et~al.(2006)Assun{\c{c}}{\~a}o, Neves, C{\^a}mara, and da~Costa~Freitas}]{assunccao2006efficient}
\bibinfo{author}{R.~M. Assun{\c{c}}{\~a}o}, \bibinfo{author}{M.~C. Neves}, \bibinfo{author}{G.~C{\^a}mara}, \bibinfo{author}{C.~da~Costa~Freitas},
\newblock \bibinfo{title}{Efficient regionalization techniques for socio-economic geographical units using minimum spanning trees},
\newblock \bibinfo{journal}{International Journal of Geographical Information Science} \bibinfo{volume}{20} (\bibinfo{year}{2006}) \bibinfo{pages}{797--811}.
\bibitem[{Teixeira et~al.(2019)Teixeira, Assun{\c{c}}{\~a}o, and Loschi}]{teixeira2019bayesian}
\bibinfo{author}{L.~V. Teixeira}, \bibinfo{author}{R.~M. Assun{\c{c}}{\~a}o}, \bibinfo{author}{R.~H. Loschi},
\newblock \bibinfo{title}{Bayesian space-time partitioning by sampling and pruning spanning trees.},
\newblock \bibinfo{journal}{J. Mach. Learn. Res.} \bibinfo{volume}{20} (\bibinfo{year}{2019}) \bibinfo{pages}{85--1}.
\bibitem[{Luo et~al.(2021)Luo, Sang, and Mallick}]{luo2021bayesian}
\bibinfo{author}{Z.~T. Luo}, \bibinfo{author}{H.~Sang}, \bibinfo{author}{B.~Mallick},
\newblock \bibinfo{title}{A {B}ayesian contiguous partitioning method for learning clustered latent variables},
\newblock \bibinfo{journal}{The Journal of Machine Learning Research} \bibinfo{volume}{22} (\bibinfo{year}{2021}) \bibinfo{pages}{1748--1799}.

\end{thebibliography}
\end{document}